\documentclass[11pt]{article}
\usepackage[utf8]{inputenc}
\usepackage[margin=1in]{geometry}
\usepackage{amsmath, mdframed}
\usepackage{amssymb}
\usepackage{amsfonts}
\usepackage{newtxtext, newtxmath}
\usepackage{enumitem}
\usepackage{titling}
\usepackage{hyperref}
\usepackage{epsfig}

\usepackage{amsthm}
\usepackage{graphicx}
\usepackage[utf8]{inputenc}
\usepackage[english]{babel}
\usepackage{CJKutf8}
\usepackage{url}
\usepackage{comment}
\usepackage{bbm}
\usepackage{blindtext}
\usepackage{xcolor}
\usepackage{tikz, tikz-3dplot}
\usepackage{tikz-qtree}
\usetikzlibrary{positioning,calc}
\usepackage{etoolbox}
\usepackage{enumitem}
\usepackage{multirow}
\usepackage{tabularx}
\usepackage{makecell}
\usetikzlibrary{decorations.markings}
\usetikzlibrary{decorations.text}

\newtheorem{theorem}{Theorem}

\newtheorem{lemma}[theorem]{Lemma}
\newtheorem{corollary}[theorem]{Corollary}
\newtheorem{prop}[theorem]{Proposition}
\newtheorem{remark}{Remark}

\newcommand{\rounddown}[1]{\left\lfloor#1\right\rfloor}
\newcommand{\roundup}[1]{\left\lceil#1\right\rceil}
\newcommand{\Z}{\mathbb{Z}}
\newcommand{\R}{\mathbb{R}}

\makeatletter

\DeclareRobustCommand{\cev}[1]{%
  {\mathpalette\do@cev{#1}}%
}
\newcommand{\do@cev}[2]{%
  \vbox{\offinterlineskip
    \sbox\z@{$\m@th#1 x$}%
    \ialign{##\cr
      \hidewidth\reflectbox{$\m@th#1\vec{}\mkern4mu$}\hidewidth\cr
      \noalign{\kern-\ht\z@}
      $\m@th#1#2$\cr
    }%
  }%
}
\makeatother

\newcommand{\knote}[1]{{\color{red}[{Karthik: \bf #1}]\marginpar{\color{red}*}}}

\newcommand{\snote}[1]{{\color{cyan}[{Siyue: \bf #1}]\marginpar{\color{blue}*}}}

\DeclareMathOperator{\supp}{supp}
\DeclareMathOperator{\opt}{OPT}

\title{Minimum Cost Nowhere-zero Flows and Cut-balanced Orientations}
\author{Karthekeyan Chandrasekaran\thanks{Grainger College of Engineering, University of Illinois, Urbana-Champaign. Email: {\tt karthe@illinois.edu}. Supported in part by NSF grant CCF-2402667.}\and Siyue Liu\thanks{Tepper School of Business, Carnegie Mellon University. Email: {\tt siyueliu@andrew.cmu.edu}. This material is based upon work supported in part by the Air Force Office of
Scientific Research under award number FA9550-23-1-0031.}
\and R. Ravi\thanks{Tepper School of Business, Carnegie Mellon University. Email: {\tt ravi@andrew.cmu.edu}. This material is based upon work supported in part by the Air Force Office of
Scientific Research under award number FA9550-23-1-0031.}
}

\date{}%\today}

\begin{document}
\maketitle

\begin{abstract}
Flows and colorings are disparate concepts in graph algorithms---the former is tractable while the latter is intractable. Tutte \cite{tutte1954contribution, tutte1966algebraic} introduced the concept of nowhere-zero flows to unify these two concepts. Jaeger \cite{jaeger1976balanced} showed that nowhere-zero flows are equivalent to cut-balanced orientations. 
Motivated by connections between nowhere-zero flows, cut-balanced orientations, Nash-Williams' well-balanced orientations, and postman problems, we study optimization versions of nowhere-zero flows and cut-balanced orientations. 
Given a bidirected graph with asymmetric costs on two orientations of each edge, we study the min cost nowhere-zero $k$-flow problem and min cost $k$-cut-balanced orientation problem. We show that both problems are NP-hard to approximate within any finite factor. Given the strong inapproximability result, we design bicriteria approximations for both problems: we obtain a $(6,6)$-approximation to the min cost nowhere-zero $k$-flow and a $(k,6)$-approximation to the min cost $k$-cut-balanced orientation.
For the case of symmetric costs (where the costs of both orientations are the same for every edge), we show that the nowhere-zero $k$-flow problem remains NP-hard and admits a $3$-approximation. 
\end{abstract}

\newpage
\section{Introduction}\label{sec:intro}
Flows and colorings are disparate concepts in graph theory, especially from a computational perspective---flow problems are typically tractable while coloring problems are typically intractable. However, these two concepts are related in planar graphs via planar duality. 
Inspired by this connection, Tutte \cite{tutte1954contribution,tutte1966algebraic} introduced nowhere-zero flows to unify flows and colorings in arbitrary graphs. This unified viewpoint has driven several fundamental results in connectivity and orientations (e.g., see \cite{west2001introduction} and \cite{CLR24}). Existence of nowhere zero flows is characterized by the existence of cut-balanced orientations \cite{jaeger1976balanced}. Motivated by these connections, we investigate optimization problems associated with nowhere-zero flows and cut-balanced orientations.

We begin by defining nowhere-zero flows, cut-balanced orientations, and the optimization problems of interest. 
%The notion of nowhere-zero flows is interesting only in $2$-edge-connected graphs, so we will assume that all graphs of interest in this work are $2$-edge-connected.
In an undirected graph $G=(V,E)$, for a subset of vertices $U\subseteq V$, we denote the set of edges with exactly one end-vertex in $U$ by $\delta_G(U)$. In a directed graph $D=(V,A)$, for $U\subseteq V$, we denote the arcs leaving and entering $U$ by $\delta^+_{D}(U)$ and $\delta^-_{D}(U)$ respectively. For a subset of edges $F\subseteq E$, we denote $\delta_F(U):=\delta_G(U)\cap F$. Similarly, for a subset of arcs $B\subseteq A$, we define $\delta_{B}^{\pm}(U):=\delta_{D}^{\pm}(U)\cap B$.
Let $G=(V, E)$ be an undirected graph and $k\ge 2$ be an integer.
A \emph{nowhere-zero $k$-flow} in $G$ is a tuple $(\vec{E},f)$, where $\vec{E}$ is an orientation of $E$ and $f:\vec{E} \rightarrow \{1,2,...,k-1\}$ is a function that satisfies flow conservation, i.e., $f(\delta_{\vec{E}}^+(v))=f(\delta_{\vec{E}}^-(v))$ for every $v\in V$. A \emph{nowhere-zero flow} is a nowhere-zero $k$-flow for some finite integer $k\geq 2$. An orientation $\vec{E}$ of $E$ \emph{induces} a nowhere-zero $k$-flow if there is a function $f: \vec{E} \rightarrow \{1,2,...,k-1\}$ such that $(\vec{E}, f)$ is a nowhere-zero $k$-flow. 
An orientation $\vec{E}$ of $E$ is \emph{$k$-cut-balanced} if $|\delta_{\vec{E}}^+(U)|\geq \frac{1}{k}|\delta_{E}(U)| \ \forall U\subseteq V$. 
A \emph{cut-balanced orientation} is a $k$-cut-balanced orientation for some finite integer $k\ge 2$. 
The \emph{bidirected graph} of $G=(V,E)$, denoted $\vec{G}=(V,E^+\cup E^-)$, is obtained by making two copies $e^+\in E^+$ and $e^-\in E^-$ of every edge $e\in E$ and orienting them in opposite directions. Let  $c:E^+\cup E^-\rightarrow \Z_{\ge 0}$ be costs on the arcs of the bidirected graph $\vec{G}$. The cost of an orientation $\vec{E}$ is denoted by $c(\vec{E}):=\sum_{e\in \vec{E}}c(e)$ and the cost of a nowhere-zero flow $(\vec{E}, f)$ is denoted by $c(f):=\sum_{e\in \vec{E}}c(e)f(e)$. 
We consider the following problems for a fixed integer $k\ge 2$:

\begin{mdframed}
\textbf{Weighted nowhere-zero $k$-flow ($\operatorname{WNZF}(k)$)}:  \\
\begin{tabular}{ l l }
\textbf{Given:}& Undirected $2$-edge-connected graph $G=(V, E)$ and costs $c: E^+\cup E^-$ \\
% \textbf{Goal:} & $\min\{c(f): \exists\ \vec{E} \text{ such that $(\vec{E}, f)$ is a nowhere-zero $k$-flow}\}$
\textbf{Goal:}& Find a nowhere-zero $k$-flow $(\vec{E},f)$ with minimum cost $c(f)$.
\end{tabular}
\end{mdframed}

\begin{mdframed}
\textbf{Weighted $k$-cut-balanced orientations ($\operatorname{WCBO}(k)$)}:\\
\begin{tabular}{ l l }
\textbf{Given:}& Undirected $2$-edge-connected graph $G=(V, E)$ and costs $c: E^+\cup E^-$ \\
% \textbf{Goal:} & $\min\{c(\vec{E}): \text{$\vec{E}$ is a $k$-cut-balanced orientation}\}$
\textbf{Goal:}& Find a $k$-cut-balanced orientation $\vec{E}$ with minimum cost $c(\vec{E})$.
\end{tabular}
\end{mdframed}

% \knote{We note that there does not exist a nowhere-zero $1$-flow or a $1$-cut-balanced orientation by definition, so the interesting range of $k$ for both problems is $\Z_{\ge 2}$}. 
The \emph{feasibility versions} of  $\operatorname{WNZF}(k)$ and $\operatorname{WCBO}(k)$ are the problems of determining whether a given undirected $2$-edge-connected graph has a nowhere-zero $k$-flow and a $k$-cut-balanced orientation, respectively. An orientation induces a nowhere-zero $k$-flow if and only if it is $k$-cut-balanced \cite{jaeger1976balanced, goddyn2001open,goddyn1998k,thomassen2012weak}. Hence, the feasibility versions of  $\operatorname{WNZF}(k)$ and $\operatorname{WCBO}(k)$ are equivalent, and we will discuss their complexity shortly. 
We say that the cost function $c: E^+\cup E^-$ is \emph{symmetric} if $c(e^+)=c(e^-)$ for every $e\in E$ and \emph{asymmetric} otherwise. If the costs $c$ are symmetric, $\operatorname{WCBO}(k)$ is equivalent to the feasibility problem while $\operatorname{WNZF}(k)$ is non-trivial. We denote $\operatorname{WNZF}(k)$ for symmetric costs as $\operatorname{SWNZF}(k)$. We define $\operatorname{WNZF}(\infty)$ as the problem of finding a min-cost nowhere-zero flow with no restriction on $k$, i.e., 
\[
\operatorname{WNZF}(\infty):=\min\{c(f): \exists\ \vec{E} \text{ such that $(\vec{E}, f)$ is a nowhere-zero $k$-flow for some integer $k\ge 2$}\}. 
\]
We define $\operatorname{WCBO}(\infty)$  and $\operatorname{SWNZF}(\infty)$ analogously. 
%We denote $\operatorname{WNZF}(k=\infty)$, $\operatorname{SWNZF}(k=\infty)$, and $\operatorname{WCBO}(k=\infty)$ as $\operatorname{WNZF}$, $\operatorname{SWNZF}$, and $\operatorname{WCBO}$ respectively. 
In all optimization problems above, we assume that the input graph is $2$-edge-connected since this is a necessary condition for the existence of a nowhere-zero flow/cut-balanced orientation.

\subsection{Background, Motivations, and Connections}
%\paragraph*{Nowhere-zero Flows.} 
\vspace{1mm}
\noindent \textbf{Nowhere-zero Flows.}
Nowhere-zero flow is a rich area of study in graph theory, with close connections to graph coloring and chromatic polynomials. 
%A central conjecture of Tutte \cite{tutte1954contribution,tutte1966algebraic} 
We begin by observing that $2$-edge-connectivity is necessary and sufficient for the existence of a nowhere-zero flow.\footnote{The necessity of $2$-edge-connectivity for nowhere-zero flows follows from flow-conservation and nowhere-zero property. We sketch its sufficiency: we recall that every $2$-edge-connected graph $G=(V, E)$ admits a strongly connected orientation $\vec{E}$, i.e., an orientation $\vec{E}$ such that $|\delta^+_{\vec{E}}(U)|\ge 1$ for every $\emptyset\neq U\subsetneq V$; such an orientation is a $k$-cut-balanced orientation for some sufficiently large integer $k$; now, recall that an orientation is a $k$-cut-balanced orientation if and only if it induces a nowhere-zero $k$-flow~\cite{jaeger1976balanced}.} 
Moreover, for integers $k_2\ge k_1\ge 2$, a nowhere-zero $k_1$-flow is also a nowhere-zero $k_2$-flow. 
Given this, most works on nowhere-zero flows focused on the least integer $k$ for which every $2$-edge-connected graph admits a nowhere-zero $k$-flow. 
% The Petersen graph is $2$-edge-connected and does not have a nowhere-zero $4$-flow. 
Tutte \cite{tutte1954contribution,tutte1966algebraic} conjectured that every $2$-edge-connected graph has a nowhere-zero $5$-flow---we will call this Tutte's $5$-flow conjecture. 
The Petersen graph does not have a nowhere-zero $4$-flow, showing that $5$ is best possible.
Jaeger \cite{jaeger1976nowhere,jaeger1979flows} proved that every $2$-edge-connected graph has a nowhere-zero $8$-flow. Seymour \cite{seymour1981nowhere} improved Jaeger's result by showing that every $2$-edge-connected graph has a nowhere-zero $6$-flow. Seymour's proof is non-constructive and a subsequent algorithmic proof was given by Younger \cite{younger1983integer}. DeVos and Nurse \cite{devos2023short} gave a shorter proof of Seymour's result which can also be made constructive. 

We now discuss the feasibility variants of $\operatorname{WNZF}(k)$ and $\operatorname{WCBO}(k)$. A consequence of Seymour's result is that for every integer $k\ge 6$, all instances of both $\operatorname{WNZF}(k)$ and $\operatorname{WCBO}(k)$ are feasible. We now discuss the status for integers $k\in \{2,3, 4, 5\}$. 
An undirected graph admits a nowhere-zero $2$-flow if and only if it is Eulerian; we recall that Eulerian property is verifiable in polynomial time and hence, the feasibility versions of $\operatorname{WNZF}(2)$ and $\operatorname{WCBO}(2)$ are solvable in polynomial time. A planar graph admits a nowhere-zero $k$-flow if and only if its dual is $k$-vertex-colorable \cite{jaeger1988nowhere}; deciding $3$-vertex-colorability of a planar graph is NP-complete \cite{garey1974some} and consequently, feasibility variants of $\operatorname{WNZF}(3)$ and $\operatorname{WCBO}(3)$ are NP-complete (even in planar graphs) \cite{kochol1998hypothetical}. 
Tutte \cite{tutte1954contribution} showed that a simple cubic graph has a nowhere-zero $4$-flow if and only if it is $3$-edge-colorable. Deciding $3$-edge-colorability of a given simple cubic graph is NP-complete \cite{holyer1981np} and hence, feasibility variants of $\operatorname{WNZF}(4)$ and $\operatorname{WCBO}(4)$ are also NP-complete (see also \cite{jaeger1979flows} and Section 7.3 of \cite{west2001introduction}). Interestingly, Kochol \cite{kochol1998hypothetical} proved that if Tutte's $5$-flow conjecture is false, then deciding whether a cubic graph has a nowhere-zero $5$-flow is NP-complete. In fact, for $k=3,4,5$, NP-completeness of the feasibility variants implies no finite approximations for corresponding optimization problems $\operatorname{WNZF}(k)$ and $\operatorname{WCBO}(k)$. This is because a finite approximation algorithm for the instance with all costs set to 0 would also certify feasibility.
We summarize these facts in the 
%following theorem. 
first two columns of Table \ref{table:main}.

From the discussion above, we have that 
the feasibility variants of $\operatorname{WNZF}(2)$ and $\operatorname{WCBO}(2)$ are polynomial-time solvable and for $k\ge 6$, all instances of $\operatorname{WNZF}(k)$ and $\operatorname{WCBO}(k)$ are feasible. Given this status, it is natural to ask whether the optimization variants $\operatorname{WNZF}(k)$ and $\operatorname{WCBO}(k)$ are solvable in polynomial time for $k=2$ and $k\ge 6$. We observe that $\operatorname{WNZF}(2)$ and $\operatorname{WCBO}(2)$ are equivalent since a nowhere-zero $2$-flow $(\vec{E},f)$ has $f(e)=1\ \forall e\in\vec{E}$, and thus the cost of the flow is equal to the cost of the orientation. 
%c(f)=c(\vec{E})$. 
In particular, $\operatorname{WNZF}(2)$ and $\operatorname{WCBO}(2)$ are equivalent to the min cost Eulerian orientation problem, which can be solved in polynomial time \cite{win1989windy}. 
We emphasize that these connections already highlight the discrete nature of $\operatorname{WNZF}(k)$ in contrast to the classic min cost flow problem which is inherently a continuous optimization problem. 

\vspace{1mm}
\noindent \textbf{Well-balanced Orientations.}
The notion of $k$-cut-balanced orientation is closely related to that of well-balanced orientation, a fundamental notion in connectivity-preserving orientation problems. Let $G=(V, E)$ be an undirected graph. An orientation $\vec{E}$ is \emph{well-balanced} if $\lambda_{\vec{E}}(u,v)\ge \lfloor \lambda_E(u,v)/2\rfloor$ for every $u,v\in V$, where 
%$\lambda_E(u,v)$ is the $\{u,v\}$-connectivity in the undirected graph $G$ and 
$\lambda_E(u,v):=\min_{u\in U\subseteq V-v}|\delta_E(U)|$ and 
%$\lambda_{\vec{E}}(u,v)$ is the directed $(u,v)$-connectivity in the digraph $(V, \vec{E})$. 
$\lambda_{\vec{E}}(u,v):=\min_{u\in U\subseteq V-v}|\delta_{\vec{E}}^+(U)|$. 
Nash-Williams \cite{nash1960orientations} showed that every $2$-edge-connected graph admits a well-balanced orientation---this result is known as the strong orientation theorem in the literature. The strong orientation theorem guarantees a single orientation that halves the pairwise connectivity of all pairs. In contrast, Nash-Williams' weak orientation theorem is about global connectivity: it states that every $2k$-edge-connected graph admits a $k$-arc-connected orientation. The weak orientation theorem can be derived from the strong orientation theorem. Nash-Williams' proof of the strong orientation theorem is rather involved. Subsequent works \cite{mader1978reduction,frank1993applications} have given alternative proofs, but they are still complex. In fact, discovering a simpler proof of Nash-Williams' strong orientation theorem has been a long-standing question in graph theory and combinatorial optimization (see Section 9.8 in \cite{frank2011connections}). 

As a motivating question for this quest for simplification, Frank \cite{frank1993applications} posed an optimization variant of Nash-Williams' strong orientation theorem, namely the min-cost well-balanced orientation problem: given an undirected graph with costs on both orientations of each edge, find a min cost well-balanced orientation. Bern\'ath, Iwata, Kir\'{a}ly, Kir\'{a}ly, and Szigeti \cite{bernath2008recent} proved that this problem is NP-hard. Further, Bern\'ath and Joret \cite{bernath2008well} proved that deciding if a mixed graph has a well-balanced orientation is also NP-hard. These results rule out the possibility of approximating the min cost well-balanced orientation within any finite factor. On the other hand, the optimization variant of Nash-Williams' weak orientation theorem---namely, given an undirected $2k$-edge-connected graph with costs on both orientations of each edge, finding a min-cost $k$-arc-connected orientation---is solvable in polynomial time via Edmonds-Giles' theory of submodular flows \cite{edmonds1977min}. 

We observe that $k$-cut-balanced orientations are closely related to well-balanced orientations: if $\vec{E}$ is a $k$-cut-balanced orientation, then for every pair $(u,v)\in V$, we have that 
\[\lambda_{\vec{E}}(u,v)=\min\limits_{u\in U\subseteq V-v}|\delta_{\vec{E}}^+(U)|\geq \min\limits_{u\in U\subseteq V-v}\frac{1}{k}|\delta_E(U)|=\frac{1}{k}\lambda_E(u,v).
\]
Thus, a $k$-cut-balanced orientation can be viewed as an ``approximately'' well-balanced orientation with the approximation factor being $k/2$. Since there exists an efficient algorithm to find a $6$-cut-balanced orientation in every $2$-edge-connected graph, it follows that there exists an efficient algorithm to find a ``$3$-approximately'' well-balanced orientation. More generally, a $k$-cut-balanced orientation with small cost  is a ``$k/2$-approximately'' well-balanced orientation with small cost, thus relating $\operatorname{WCBO}(k)$ to a bicriteria version of Frank's problem. This connection was one of the motivations for us to study the optimization version of $k$-cut-balanced orientation problem, namely $\operatorname{WCBO}(k)$.

\vspace{1mm}
\noindent \textbf{Orientation and Postman problems.}
We observe that $\operatorname{WCBO}(\infty)$ is polynomial-time solvable. In fact, it is equivalent to the problem of finding a min cost strongly connected orientation: 
an orientation $\vec{E}$ is $k$-cut-balanced for some finite integer $k\ge 2$ if and only if 
%for very large $k$, an orientation $\vec{E}$ is $k$-cut-balanced if and only if 
for every $\emptyset\neq U\subsetneq V$, we have that $|\delta_{\vec{E}}^+(U)|\geq \frac{1}{k}|\delta_{E}(U)|$, i.e., $|\delta_{\vec{E}}^+(U)|\geq \roundup{\frac{1}{k}|\delta_{E}(U)|}$, i.e., $|\delta_{\vec{E}}^+(U)|\geq 1$. Min cost strongly connected orientation can be solved in polynomial time (via \cite{lucchesi1978minimax, edmonds1977min}). 

%To the best of our knowledge, $\operatorname{WNZF}$ has not been studied in the literature. 
In $\operatorname{WNZF}(\infty)$, we seek a nowhere-zero flow $(\vec{E}, f)$ which is equivalent to an integer-valued flow $f: E^+\cup E^-\rightarrow \Z_{\ge 0}$ such that the flow is positive in {\it exactly} one of $e^+$ and $e^-$ for every edge $e\in E$; instead, if we seek an integer-valued flow $f: E^+\cup E^-\rightarrow \Z_{\ge 0}$ such that {\it at least} one of $e^+$ and $e^-$ has positive flow value for every $e\in E$, then the resulting problem is the \emph{asymmetric postman problem} \cite{minieka1979chinese}: find a min cost directed Eulerian tour of the graph $G$ that traverses every edge $e\in E$ at least once. 
%We note that the latter problem allows solution arcs in both directions of the same edge without needing to cancel them out as would be required when interpreted as a flow. Moreover, we can convert every nowhere-zero flow into an Eulerian tour that traverses $\{e^+,e^-\}$ at least once for every $e\in E$ by decomposing the flow into cycles and stitching them together in a tour. 
In contrast to the symmetric postman problem which can be solved in polynomial time \cite{edmonds1973matching}, 
%the asymmetric postman problem is NP-hard \cite{papadimitriou1976complexity,guan1984windy}. 
% and admits a $3/2$-approximation \cite{win1989windy, raghavachari1999approximation}. 
%\snote{Win \cite{win1989windy} gave a $2$-approximation algorithm for the problem.  Raghavachari and Veerasamy \cite{raghavachari1999approximation} improved the approximation ratio to $\frac{3}{2}$. A well-studied special case of asymmetric postman is the \emph{mixed postman} problem, where we are given a mixed graph with costs on the arcs and \textit{symmetric} costs on the edges. We are asked to find a min cost Eulerian tour that traverses every arc and edge at least once. This is already an NP-hard problem \cite{papadimitriou1976complexity} and several approximation algorithms are given \cite{edmonds1973matching,christofides1984optimal,frederickson1979approximation,raghavachari19993}, with current best approximation ratio $\frac{3}{2}$ given by \cite{raghavachari19993}. The mixed postman problem with restrictions on arcs or edges has been studied by \cite{veerasamy1999approximation,zaragoza2003postman,martinez2006complexity,zaragoza2016approximation}), where one can only traverse every arc or edge exactly once, resp. The former problem is NP-hard \cite{zaragoza2003postman} and admits a $\frac{4}{3}$-approximation algorithm \cite{zaragoza2003postman}. The latter problem is also NP-hard \cite{veerasamy1999approximation,martinez2006complexity} with no approximation guarantee known.}
the asymmetric postman problem is already NP-hard but admits a $3/2$-approximation \cite{win1989windy, raghavachari1999approximation}. A well-studied special case of the asymmetric postman problem is the \emph{mixed postman} problem, where we are given a mixed graph (consisting of undirected edges and directed arcs) with costs on the arcs and \textit{symmetric} costs on the edges and the goal is to find a min cost Eulerian tour that traverses every arc and edge at least once. This is also NP-hard \cite{papadimitriou1976complexity} and several works have focused on its approximability \cite{edmonds1973matching,christofides1984optimal,frederickson1979approximation,raghavachari19993}---the current best approximation ratio is $\frac{3}{2}$ \cite{raghavachari19993}. 
The mixed postman problem with the restriction that each arc/edge can be traversed at most once has also been studied. The arc-restricted problem is NP-hard and admits a $4/3$-approximation~\cite{zaragoza2003postman}, while the edge-restricted problem is NP-hard to approximate within any finite factor \cite{zaragoza2003postman,veerasamy1999approximation,martinez2006complexity}.

$\operatorname{WNZF}(\infty)$ can be viewed as the asymmetric postman problem with \emph{orientation constraints}: i.e., flow is allowed to be positive on exactly one of the two orientations of each edge.
We will subsequently take this viewpoint while designing LP-based approximation algorithms for $\operatorname{WNZF}(k)$. We remark that $\operatorname{WNZF}(\infty)$, i.e., postman problem with orientation constraints, is different from the (mixed) postman problem with the restriction that each edge can be traversed at most once---the latter requires the flow value to be exactly $1$ on one orientation of each edge and $0$ on the other while $\operatorname{WNZF}(\infty)$ requires the flow value to be positive on exactly one orientation of each edge. Orientation constraints requiring flow value to be positive on exactly one orientation of each edge arise naturally in network design problems since they model one-way roads. 
We mention that several prior works have considered orientation constraints in directed network design problems and most of them have concluded that orientation constraints tend to make the problem much harder (e.g., see \cite{frank2003combined,khanna2005directed,singh2018approximating,cygan2013steiner}). 
\subsection{Our Results}
We first show hardness and inapproximability of $\operatorname{WNZF}(k)$ and $\operatorname{WCBO}(k)$. 
\begin{theorem}\label{thm:hardness-WNZF-WCBO}
    For every finite integer $k\ge 3$, $\operatorname{WNZF}(k)$ and $\operatorname{WCBO}(k)$ are NP-hard to approximate within any finite factor.
\end{theorem}

Theorem \ref{thm:hardness-WNZF-WCBO} rules out the possibility of approximation algorithms for $\operatorname{WNZF}(k)$ and $\operatorname{WCBO}(k)$. Given this status, we investigate bicriteria approximation algorithms for both these problems. We say that an algorithm is an \emph{$(\alpha,\beta)$-approximation} for $\operatorname{WNZF}(k)$ (resp. $\operatorname{WCBO}(k)$) if the algorithm returns a nowhere-zero $\beta k$-flow (resp. $\beta k$-cut-balanced orientation) with cost at most $\alpha c(\opt)$ where $c(\opt)$ is the minimum cost of a nowhere-zero $k$-flow (resp. $k$-cut-balanced orientation). We show the following bicriteria approximation for $\operatorname{WNZF}(k)$. 
\begin{theorem}\label{thm:approx-nwz-kflow}
For every finite integer $k\ge 6$ and for $k=\infty$, $\operatorname{WNZF}(k)$ admits a $(6,6)$-approximation. 
\end{theorem}
We observe that Theorem \ref{thm:approx-nwz-kflow} implies a $6$-approximation for $\operatorname{WNZF}(\infty)$ as stated below. 
%A $6$-approximation algorithm for $\operatorname{WNZF}$ follows immediately from the $(6,6)$-approximation algorithm for $\operatorname{WNZF}(k)$ by letting $k=\infty$. \knote{This seems to contradict Theorem \ref{thm:hardness-WNZF-WCBO}?}
\begin{corollary}\label{cor:approx-nwz-flow}
$\operatorname{WNZF}(\infty)$ admits a $6$-approximation. 
\end{corollary}
% \knote{We contrast Corollary  \ref{cor:approx-nwz-flow} with Theorem \ref{thm:hardness-WNZF-WCBO}: Theorem \ref{thm:hardness-WNZF-WCBO} is an inapproximability result for $\operatorname{WNZF}(k)$ for every finite integer $k\ge 3$ while Corollary  \ref{cor:approx-nwz-flow} is an approximability result for $\operatorname{WNZF}(\infty)$.}

Next, we turn to bicriteria approximation for $\operatorname{WCBO}(k)$. We recall that $\operatorname{WCBO}(\infty)$ is equivalent to the min cost strongly connected orientation problem which is solvable in polynomial time \cite{younger1983integer, edmonds1977min}. Hence, we focus on $\operatorname{WCBO}(k)$ for finite integers $k\ge 3$. 
Moreover, Theorem \ref{thm:hardness-WNZF-WCBO} implies that $\operatorname{WCBO}(k)$ is inapproximable for every finite integer $k\ge 3$. We complement these results by showing the following bicriteria approximation. 
%We show the following bicriteria approximation for $\operatorname{WCBO}(k)$. 
\begin{theorem}\label{thm:approx-cut-balanced}
For every integer $k\ge 6$, $\operatorname{WCBO}(k)$ admits a $(k,6)$-approximation. 
\end{theorem}

In most applications of nowhere-zero $k$-flows and $k$-cut-balanced orientations, we note that $k$ is a small constant. In particular, Theorems \ref{thm:approx-nwz-kflow} and \ref{thm:approx-cut-balanced} imply $(6,6)$-approximation algorithms for min cost nowhere-zero $6$-flow and min cost $6$-cut-balanced orientation respectively. 

Next, we turn to the symmetric cost variants of the problems. 
As mentioned earlier, the symmetric cost variant of $\operatorname{WCBO}(k)$ is equivalent to the feasibility variant, so we focus only on the symmetric cost variant of $\operatorname{WNZF}(k)$. We recall that $\operatorname{SWNZF}(\infty)$ is closely related to the symmetric postman problem: it can be viewed as an orientation constrained min-cost symmetric postman problem. While the min-cost symmetric postman problem is polynomial-time solvable, we show that $\operatorname{SWNZF}(\infty)$ is NP-hard even for unit costs. We show the following hardness and inapproximability results for $\operatorname{SWNZF}(k)$. 
\begin{theorem}\label{thm:complexity-SWNZF-k}
    For every finite integer $k\ge 3$ and for $k=\infty$, $\operatorname{SWNZF}(k)$ for unit costs is NP-hard. 
    %$\operatorname{SWNZF}(k)$ for unit costs is NP-hard for every integer $k\in [3,\infty]$. 
    Moreover, for $k=3$ and $k=4$, $\operatorname{SWNZF}(k)$ is NP-hard to approximate within any finite factor and $\operatorname{SWNZF}(5)$ is NP-hard to approximate within any finite factor assuming Tutte's $5$-flow conjecture is false.
\end{theorem}

Motivated by the hardness results in Theorem \ref{thm:complexity-SWNZF-k}, we turn to approximation algorithms for $\operatorname{SWNZF}(k)$. If the costs are symmetric, we observe that every nowhere-zero $6$-flow $f$ is a $5$-approximation for $\operatorname{SWNZF}(k)$ for every finite integer $k\ge 6$ and for $k=\infty$: this is because the flow value of $f$ on each edge is at most $5$, while the min cost nowhere-zero $k$-flow has to send at least one unit flow on each edge. We improve on this naive $5$-approximation for $\operatorname{SWNZF}(k)$ to achieve an approximation factor of $3$. We refer the reader to Table \ref{table:main} for a summary of our results.
\begin{theorem}\label{thm:approx-SWNZF-k}
For every finite integer $k\ge 6$ and for $k=\infty$, $\operatorname{SWNZF}(k)$ admits a $3$-approximation. 
\end{theorem}

\begin{table}[htbp]
\begin{center}
\begin{tabular}{|c|c|c|c|c|}
\hline
\textbf{Integer $k$} & \textbf{Feasibility}& \textbf{WNZF$(k)$} & \textbf{WCBO$(k)$} & \textbf{SWNZF$(k)$}\\
    & \textbf{NZF$(k)$} $\&$ \textbf{CBO$(k)$} & & &  
\\ \hline
2  & Poly-time [folklore] & Poly-time \cite{win1989windy} & Poly-time \cite{win1989windy} & Poly-time \cite{win1989windy} 
\\ \hline
3, 4 & NP-complete & No finite approx & No finite approx & No finite approx \\
& \cite{jaeger1988nowhere, tutte1954contribution, holyer1981np} & & &
\\ \hline
5  & NP-complete & No finite approx & No finite approx & No finite approx \\
(assuming Tutte's & \cite{kochol1998hypothetical}  & & & \\
$5$-flow conjecture & & & &\\ is false) & & & & 
\\ \hline
finite $k\ge 6$ & Always feasible & No finite approx  & No finite approx & NP-hard for unit costs\\
 & \cite{seymour1981nowhere, younger1983integer}  & (Theorem \ref{thm:hardness-WNZF-WCBO}) & (Theorem \ref{thm:hardness-WNZF-WCBO}) & (Theorem \ref{thm:complexity-SWNZF-k}) \\
 & & $(6,6)$-approx & $(k,6)$-approx & $3$-approx \\
 & & (Theorem \ref{thm:approx-nwz-kflow}) & (Theorem \ref{thm:approx-cut-balanced}) & (Theorem \ref{thm:approx-SWNZF-k}) 
\\ \hline
$\infty$ & Always feasible & & Poly-time & NP-hard for unit costs\\
 & \cite{seymour1981nowhere, younger1983integer}  & & \cite{lucchesi1978minimax,edmonds1977min} & (Theorem \ref{thm:complexity-SWNZF-k}) \\
 & & $6$-approx & & $3$-approx \\
 & & (Corollary \ref{cor:approx-nwz-flow}) & & (Theorem \ref{thm:approx-SWNZF-k})
\\ \hline
\end{tabular}
\caption{Hardness and approximation algorithms for $\operatorname{WNZF}(k)$, $\operatorname{WCBO}(k)$ and $\operatorname{SWNZF}(k)$. 
Feasibility $\operatorname{NZF}(k)$ and $\operatorname{CBO}(k)$ refer to the feasibility versions of $\operatorname{WNZF}(k)$ and $\operatorname{WCBO}(k)$ respectively.
All inapproximability results are assuming $P\neq NP$.}
\label{table:main}
\end{center}
\end{table}

\vspace{-3em}
\subsection{Techniques}
We observe that the equivalence of nowhere-zero $k$-flows and $k$-cut-balanced orientations~\cite{jaeger1976balanced} implies a reduction between $\operatorname{WNZF}(k)$ and $\operatorname{WCBO}(k)$ with a small loss in approximation factor--see Lemma \ref{lemma:approx_translate}. A consequence of this reduction is that NP-hardness of approximating $\operatorname{WNZF}(k)$ within any finite factor implies the same for $\operatorname{WCBO}(k)$ and vice-versa. With this consequence, we observe that proving Theorem \ref{thm:hardness-WNZF-WCBO} reduces to showing that $\operatorname{WCBO}(k)$ does not admit a finite approximation. We prove the latter by a reduction from the satisfiability problem. We prove the hardness result for $\operatorname{SWNZF}(k)$ mentioned in Theorem \ref{thm:complexity-SWNZF-k} by a reduction from Not-All-Equal-3-SAT. Our reductions are based on careful choice of gadgets.  

We now discuss the techniques underlying our algorithmic results. For $\operatorname{WNZF}(k)$, we write an integer program and show extreme points properties of its LP relaxation. In particular, let $x^*\in \mathbb{R}^{E^+\cup E^-}$ be an extreme point optimum solution to the LP relaxation. We show that $x^*$ is half-integral. We prove that the edges of the undirected graph can be partitioned into integral and non-integral edges: An integral edge has $x^*$ being integral in both directions; a non-integral edge has $x^*$ equal to $\frac{1}{2}$ in both directions. We show that the integral arcs in the solution already form a `partial' $k$-flow $f$ (that is conserved at every node but may be zero on some edges). Since we are interested in a nowhere-zero flow, we need to send flow along non-integral edges. 
%structural properties of an extreme point $x^*$ of its LP relaxation. In particular, we can divide the edges into integral and non-integral. An integral edge has $x^*$ being integral in both directions; a non-integral edge has $x^*$ equal to $\frac{1}{2}$ in both directions. We show that the integral arcs in the solution already form a `partial' k-flow $f$. 
To send flow on the non-integral edges, we use an arbitrary nowhere-zero $6$-flow $g$ of the whole graph which can be found in polynomial time. We then use the observation that the composed flow $6f+g$ is nonzero in every edge and has all flow values at most $6k-1$ to conclude that $6f+g$ is a nowhere-zero $6k$-flow. For the same reason, $6f-g$ is also a nowhere-zero $6k$-flow. Our algorithm returns the smaller cost one among $6f+g$ and $6f-g$. Finally since the non-integral edges have LP values equal to $\frac{1}{2}$ in both directions in the extreme point $x^*$, we can bound the approximation factor of the returned nowhere-zero $6k$-flow relative to the optimum objective value of the LP-relaxation.
% \snote{The above needs def of $f\pm 6g$}

For $\operatorname{WCBO}(k)$, the reduction mentioned in the first paragraph above (and detailed in Lemma \ref{lemma:approx_translate}) and Theorem \ref{thm:approx-nwz-kflow} imply a $(6(k-1),6)$-approximation. We improve on the factors to achieve a $(k,6)$-approximation in Theorem \ref{thm:approx-cut-balanced}. For this, 
we write an integer program and study its LP-relaxation. In particular, we show that the LP-relaxation is $1/k$-integral via a connection to the theory of submodular flows. The integral arcs of this LP-relaxation no longer form a `partial' $k$-flow. However, we show that we can always complete the orientation given by the integral arcs into a $k$-cut-balanced orientation, which induces a partial $k$-flow $f$. Thus, we can use a similar approach as for $\operatorname{WNZF}(k)$ to obtain a $6k$-cut-balanced orientation with cost at most $k$ times the cost of the LP-relaxation.

For $\operatorname{SWNZF}(k)$, we achieve a unicriteria approximation via combinatorial techniques. Our algorithm finds a \emph{locally optimal} nowhere-zero $6$-flow $(\vec{E},f)$, that is, a nowhere-zero $6$-flow whose cost cannot be reduced by pushing $6$ units of flow along any directed cycle in $\vec{E}$. 
We show that a locally optimal nowhere-zero $6$-flow gives a $3$-approximation to $\operatorname{SWNZF}(k)$ for every finite integer $k\geq 6$ and for $k=\infty$.
%such that we cannot push $6$ unit flow along any directed cycle in $\vec{E}$ and reduce the cost of $f$. 
Next, we turn to the question of existence and an efficient algorithm to find a locally optimal nowhere-zero $6$-flow. 
We show that for an arbitrary nowhere-zero $6$-flow $f$, there is a locally optimal nowhere-zero $6$-flow $f'$ that can be obtained starting from $f$ and repeatedly pushing $6$ units of flow along the reverse directions of directed cycles; this proves the existence of a locally optimal nowhere-zero $6$-flow. We phrase the problem of finding such a locally optimal nowhere-zero $6$-flow via such cycle augmentations from an arbitrary $f$ as a min-cost circulation problem. This can be solved in polynomial time, and its optimum can be used to recover a locally optimal nowhere-zero $6$-flow. 

%\paragraph*{Organization} 
\vspace{1mm}
\noindent \textbf{Organization.}
In Section \ref{sec:basic}, we setup some notation, present basic properties of nowhere-zero $k$-flows, and show the equivalence of nowhere-zero $k$-flows and $k$-cut-balanced orientations. In Section \ref{sec:reduction}, we give the reduction between approximation algorithms for $\operatorname{WNZF}(k)$ and $\operatorname{WCBO}(k)$. In Section \ref{sec:complexity-WNZF+WCBO}, we prove the hardness of $\operatorname{WNZF}(k)$ and $\operatorname{WCBO}(k)$. In Section \ref{sec:algo-WNZF+WCBO}, we give LP relaxations and design bicriteria approximation algorithms for $\operatorname{WNZF}(k)$ and $\operatorname{WCBO}(k)$. In Section \ref{sec:SWNZF}, we prove the hardness and give unicriteria approximation algorithms for $\operatorname{SWNZF}(k)$.

\section{Preliminaries}\label{sec:basic}
Let $G=(V, E)$ be an undirected graph. Let $\vec{G}=(V,E^+\cup E^-)$ denote its bidirected graph. A \emph{partial orientation} is an orientation $\vec{F}$ on a subset $F\subseteq E$ of edges. A \emph{partial $k$-flow} (or \emph{$k$-flow}) is a partial orientation $\vec{F}$ and a function $f:\vec{F} \rightarrow \{1,2,...,k-1\}$ such that $f(\delta_{\vec{F}}^+(v))=f(\delta_{\vec{F}}^-(v))$. Denote the reverse orientation of an arc $e$ by $e^{-1}$. For a partial orientation $\vec{F}$, denote the reverse orientation by $\cev{F}$. Given a $k$-flow $(\vec{F},f)$, we can extend the domain of the function $f$ to the arcs $E^+\cup E^-$ of the bidirected graph $\vec{G}$ by defining $f(e):=-f(e^{-1}) \ \forall e\in \cev{F}$ and $f(e):=0 \ \forall e\notin \vec{F}\cup \cev{F}$. We call the resulting extended function the \emph{extension} of the partial $k$-flow $f$ to $E^+\cup E^-$. We note that the extension of a partial $k$-flow satisfies flow conservation. This gives us the following equivalent definition of a $k$-flow: a $k$-flow is a function $f:E^+\cup E^-\rightarrow \{0, \pm 1,...,\pm (k-1)\}$ such that $f(e^+)=-f(e^-)\ \forall e\in E$ and $f(\delta_{\vec{G}}^+(v))=f(\delta_{\vec{G}}^-(v))$. A $k$-flow $f$ is a nowhere-zero $k$-flow if $f(e)\neq 0\ \forall e\in E^+\cup E^-$. For a $k$-flow $f:E^+\cup E^-\rightarrow \{0, \pm 1,...,\pm (k-1)\}$, we denote $\supp(f):=\{e\in E: f(e^+)\neq 0\}$ as the edges oriented by $f$ and $\supp^+(f):=\{e\in E^+\cup E^-: f(e)>0\}$ as the partial orientation associated with $f$. One can verify the two definitions of (nowhere-zero) $k$-flows are equivalent by letting $\vec{F}=\supp^+(f)$. 
Due to the equivalence of the two definitions, we will sometimes refer to a $k$-flow $(\vec{F}, f)$ by the extension $f:E^+\cup E^-\rightarrow \{0, \pm 1,...,\pm (k-1)\}$, omitting the orientation $\vec{F}$.

Equipped with this alternative definition, we can talk about negation, scaling and summation of $k$-flows. Let $G=(V, E)$ be an undirected graph. For a $k$-flow $(\vec{F},f)$, $-f$ is defined as the negation of the extension of $f$ to $E^+\cup E^-$. We observe that $-f$ is also a $k$-flow and moreover, $\supp^+(-f)=\cev{F}$. Let $(\vec{E}_1,f_1)$ be a $k_1$-flow and $(\vec{E}_2,f_2)$ be a $k_2$-flow of $G$ for some subsets $E_1,E_2\subseteq E$. The sum of the two flows $f=f_1+f_2$ is defined as the sum of the extensions of $f_1$ and $f_2$ to $E^+\cup E^-$. Let $\vec{E}:=\supp^+(f)$ be the partial orientation associated with $f$. We note that $f$ may not be a nowhere-zero flow even if $E_1\cup E_2=E$ because of flow cancellation between $f_1$ and $f_2$. 
% \knote{if $E_1 \cup E_2 = E$, then the flow $f=k_2 f_1+f_2$ is nowhere-zero: indeed, ...} 
However, given $E_1\cup E_2=E$, if we scale $f_1$ by a factor of $k_2$, the resulting flow $f=k_2 f_1+f_2$ is nowhere-zero. This observation has been very useful in constructing nowhere-zero $k$-flows in general graphs (e.g. \cite{jaeger1976nowhere,jaeger1979flows,seymour1981nowhere}). We summarize it in the following proposition and give its proof for completeness.

\begin{prop}\label{prop:sum-flows}
    Let $G=(V,E)$ be an undirected graph and $E_1,E_2\subseteq E$ be such that $E_1\cup E_2=E$. Let $(\vec{E}_1,f_1)$ be a $k_1$-flow and $(\vec{E}_2,f_2)$ be a $k_2$-flow. Then, $f=k_2f_1+f_2$ is a nowhere-zero $k_1k_2$-flow. Moreover, $\vec{E}:=\supp^+(f)$ has the same orientation as $\vec{E}_1$ on $E_1$ and the same orientation as $\vec{E}_2$ on $E_2\setminus E_1$.
\end{prop}
\begin{proof}
    For each $e\in \supp^+(f_1)$, since $f(e)=k_2f_1(e)+f_2(e)\geq k_2-(k_2-1)\geq 1$, $f_2$ does not cancel $k_2f_1$ to value $0$. Thus, $\vec{E}$ has the same orientation as $\vec{E}_1$ on $E_1$ and the same orientation as $\vec{E}_2$ on $E_2\setminus E_1$. Moreover, for each $e\in E^+\cup E^-$, the flow value is bounded by $f(e)=k_2f_1(e)+f_2(e)\leq k_2(k_1-1)+(k_2-1)\leq k_1k_2-1$. Therefore, $f$ is a nowhere-zero $k_1k_2$-flow.
\end{proof}

The feasibility of nowhere-zero $k$-flows and $k$-cut-balanced orientations are equivalent due to the following lemma by Jaeger \cite{jaeger1976balanced}. Since we will use this fact repeatedly, we include a proof.
% in Appendix \ref{sec:proof-equivalence}.
\begin{lemma}[Jaeger \cite{jaeger1976balanced}]\label{lemma:equivalence}
Let $G=(V, E)$ be an undirected graph and $k\ge 2$ be an integer. An orientation is $k$-cut balanced if and only if it induces a nowhere-zero $k$-flow. Moreover, given a $k$-cut balanced orientation, there exists a polynomial-time algorithm to construct a nowhere-zero $k$-flow.
\end{lemma}
\begin{proof}
% [Proof of Lemma \ref{lemma:equivalence}]
    Reverse direction: Let $\vec{E}$ and $f:\vec{E}\rightarrow \{1,...,k-1\}$ be a nowhere-zero $k$-flow. By flow conservation, $f(\delta^+_{\vec{E}}(U))=f(\delta^-_{\vec{E}}(U))\ \forall U \subseteq V$. Thus, we have $1\cdot |\delta^+_{\vec{E}}(U)|\leq f(\delta^+_{\vec{E}}(U))=f(\delta^-_{\vec{E}}(U))\leq (k-1)\cdot |\delta^-_{\vec{E}}(U)|$. Similarly, we also have $1\cdot |\delta^-_{\vec{E}}(U)|\leq f(\delta^-_{\vec{E}}(U))=f(\delta^+_{\vec{E}}(U))\leq (k-1)\cdot |\delta^+_{\vec{E}}(U)|$. It follows from the equality $|\delta_E(U)|=|\delta_{\vec{E}}^+(U)|+|\delta_{\vec{E}}^-(U)|$ that $\frac{1}{k}|\delta_{E}(U)|\leq |\delta^+_{\vec{E}}(U)|\leq \frac{k-1}{k}|\delta_{E}(U)|$. Thus, $\vec{E}$ is a $k$-cut-balanced orientation.

    Forward direction: Let $\vec{E}$ be a $k$-cut-balanced orientation. By Hoffman's circulation theorem \cite{hoffman2003some}, there exists a circulation $f\in\R^{\vec{E}}$ satisfying $1\leq f(e)\leq k-1\ \forall e\in \vec{E}$, if and only if $|\delta_{\vec{E}}^-(U)|\leq (k-1)|\delta_{\vec{E}}^+(U)|\ \forall U\subsetneqq V,U\neq\emptyset$, which is equivalent to $|\delta_{\vec{E}}^+(U)|\geq \frac{1}{k}|\delta_{E}(U)|\ \forall U\subsetneqq V,U\neq\emptyset$, satisfied by the $k$-cut-balanced condition. Thus, $f$ is a nowhere-zero $k$-flow. Moreover, we can use any circulation algorithm to construct $f$ in polynomial time (see e.g. \cite{ford2015flows,chen2022maximum}).
\end{proof}

\section{Reductions between $\operatorname{WNZF}(k)$ and $\operatorname{WCBO}(k)$}\label{sec:reduction}
Lemma \ref{lemma:equivalence} implies that approximation algorithms for $\operatorname{WNZF}(k)$ and $\operatorname{WCBO}(k)$ can be translated to each other with a small loss in factors as shown in the following lemma. A consequence of the lemma below is that NP-hardness of approximating $\operatorname{WCBO}(k)$ within any finite factor implies the same for $\operatorname{WNZF}(k)$, and vice versa.
\begin{lemma}\label{lemma:approx_translate}
    If $\operatorname{WNZF}(k)$ has an $(\alpha,\beta)$-approximation algorithm, then $\operatorname{WCBO}(k)$ has a $((k-1)\alpha,\beta)$-approximation algorithm; if $\operatorname{WCBO}(k)$ has an $(\alpha,\beta)$-approximation algorithm, then $\operatorname{WNZF}(k)$ has a $((\beta k-1)\alpha,\beta)$-approximation algorithm.
\end{lemma}
\begin{proof}
    Let $(\vec{E}_1^*,f_1^*)$ be a nowhere-zero $k$-flow with minimum cost and let $\vec{E}_2^*$ be a $k$-cut-balanced orientation with minimum cost on the same instance. Let $f_2^*$ be a nowhere-zero $k$-flow induced by $\vec{E}_2^*$, which can be constructed in polynomial time according to Lemma \ref{lemma:equivalence}.

Suppose that $\operatorname{WNZF}(k)$ has an $(\alpha,\beta)$-approximation algorithm. Then, applying the algorithm returns a nowhere-zero $\beta k$-flow $(\vec{E},f)$. By Lemma \ref{lemma:equivalence}, $\vec{E}$ is a $\beta k$-cut-balanced orientation. Moreover, the orientation satisfies
    \[
    c(\vec{E})\leq c(f)\leq \alpha \cdot c(f_1^*)\leq \alpha \cdot c(f_2^*)\leq (k-1)\alpha \cdot c(\vec{E}_2^*),
    \]
    where the first inequality follows from the fact that $f$ is nowhere-zero. The second inequality follows from the fact that $f$ is an $(\alpha,\beta)$-approximate solution to $\operatorname{WNZF}(k)$. The third inequality follows from the fact that $f_1^*$ is a min cost nowhere-zero $k$-flow. The fourth inequality follows from the fact that $f_2^*$ is a nowhere-zero $k$-flow.

Suppose that $\operatorname{WCBO}(k)$ has an $(\alpha,\beta)$-approximation algorithm. Then, applying the algorithm returns a $\beta k$-cut-balanced orientation $\vec{E}$. By Lemma \ref{lemma:equivalence}, the orientation $\vec{E}$ induces a nowhere-zero $\beta k$-flow $f$. Moreover, the flow $f$ satisfies
 \[
    c(f)\leq (\beta k-1)\cdot c(\vec{E})\leq (\beta k-1)\alpha \cdot c(\vec{E}_2^*)\leq (\beta k-1)\alpha \cdot c(\vec{E}_1^*)\leq (\beta k-1)\alpha \cdot c(f_1^*),
    \]
    where the first inequality follows from the fact that $f$ is a nowhere-zero $\beta k$-flow. The second inequality follows from the fact that $\vec{E}$ is an $(\alpha,\beta)$-approximate solution to $\operatorname{WCBO}(k)$. The third inequality follows from the fact that $\vec{E}_2^*$ is a min cost $k$-cut-balanced orientation. The fourth inequality follows from the fact that $f_1^*$ is nowhere-zero.
\end{proof}

\section{Hardness of $\operatorname{WNZF}(k)$ and $\operatorname{WCBO}(k)$}\label{sec:complexity-WNZF+WCBO}
% The case of $k=2$ for both $\operatorname{WNZF}(k)$ and $\operatorname{WCBO}(k)$ is special: both problems correspond to the min-cost Eulerian orientation problem which is polynomial-time solvable \cite{win1989windy}. For $k\ge 3$, we prove that both $\operatorname{WNZF}(k)$ and $\operatorname{WCBO}(k)$ are NP-hard to approximate within any constant factor. 
We prove Theorem \ref{thm:hardness-WNZF-WCBO} in this section, i.e., we show that both $\operatorname{WNZF}(k)$ and $\operatorname{WCBO}(k)$ do not admit an algorithm with finite approximation factors for every finite $k\geq 3$.
By Lemma \ref{lemma:approx_translate}, NP-hardness of approximating $\operatorname{WCBO}(k)$ within any finite factor implies the same for $\operatorname{WNZF}(k)$, and vice versa.
Thus, it suffices to prove the NP-hardness of approximation within any finite factor for $\operatorname{WCBO}(k)$.
To achieve this, we study the problem of determining whether an arbitrary partial orientation $\vec{F}$ of a $2$-edge-connected graph $G=(V,E)$ can be completed into a $k$-cut-balanced orientation. This is the key step towards proving the hardness of $\operatorname{WCBO}(k)$. This question is known to be NP-hard for $k=3,4$, because if we take $\vec{F}=\emptyset$, this recovers the problems of deciding whether a $2$-edge-connected graph admits a nowhere-zero $3$-flow and a nowhere-zero $4$-flow, which are both NP-hard. This question is also hard for $k=5$ if Tutte's $5$-flow conjecture is false. Thus, the problem of deciding whether a partial orientation can be completed into a $k$-cut-balanced orientation is interesting only for $k\geq 6$, for which we know that the graph has a nowhere-zero $k$-flow, and possibly for $k=5$. However, we show in Theorem \ref{thm:complete-CBO-hard} that this problem is hard for every finite integer $k\geq 3$.

\emph{Satisfiability (SAT)} is a well-known NP-complete problem: the input is a collection of $n$ variables and $m$ clauses, where each clause is a disjunction of a subset of variables or their negation. The goal is to determine if there is an assignment of Boolean values to variables such that all clauses are satisfied. We call a SAT instance as a \emph{restricted SAT} if every variable appears in at most $3$ clauses. Restricted SAT is also NP-complete \cite{tovey1984simplified}.

\begin{theorem}\label{thm:complete-CBO-hard}
The following problem is NP-hard for every finite integer $k\ge 3$: given an undirected graph $G=(V, E)$ and a partial orientation $\vec{F}$ of a subset $F\subseteq E$ of edges, decide whether $\vec{F}$ can be completed into a $k$-cut-balanced orientation.
\end{theorem}
\begin{proof}
    If $k=3$, take $\vec{F}=\emptyset$. The theorem follows from the NP-completeness of determining whether a given graph has a $3$-cut-balanced orientation (see Table \ref{table:main}). 
    %by Theorem \ref{thm:complexity-feasibility}. 
    Thus, we assume $k\geq 4$ from now on.
    
    We reduce from restricted SAT where every variable appears at most $3$ times. Consider a restricted SAT instance with variables $x_1,...,x_n$ and clauses $C_1,...,C_m$. We may assume that every variable $x_i$ appears at least once positively and at least once negatively: otherwise, $x_i$ always appears positive (resp. negative), in which case we can simply assign $x_i=1$ (resp. $x_i=0$) and delete all clauses containing $x_i$. 
    
    Construct a graph $G=(V,E)$ and a partial orientation $\vec{F}$ as follows. Fix a root $r$. Introduce $2n$ nodes $\{u_1,u_1',...,u_n,u_n'\}$ and $m$ nodes $\{v_1,...,v_m\}$ such that $(u_i,v_j)\in \vec{F}$ if and only if $x_i$ appears positively in $C_j$; $(u'_i,v_j)\in \vec{F}$ if and only if $x_i$ appears negatively in $C_j$. Suppose $x_i$ appears $a_i\in\{1,2\}$ times positively and $a'_i\in\{1,2\}$ times negatively. Add one arc $(r,u_i)$ and $k-a_i-2$ arcs $(u_i,r)$ to $\vec{F}$. Similarly, add one arc $(r,u'_i)$ and $k-a'_i-2$ arcs $(u'_i,r)$ to $\vec{F}$. Note that $k-a_i-2\geq 0$ since $k\geq 4$ and $a_i\leq 2$. For each clause $C_j$, add $|C_j|+1$ arcs $(v_j,r)$ to $\vec{F}$. Finally, let $E$ be the union of the underlying undirected graph of $\vec{F}$ and $\cup_{i=1}^n\{(u_i,u_i')\}$ (see Figure \ref{fig:complexity-CBO-small} left for details and Figure \ref{fig:complexity-CBO-big} left for an example).
\begin{figure}[htbp]
    \centering
    \includegraphics[width=0.9\linewidth]{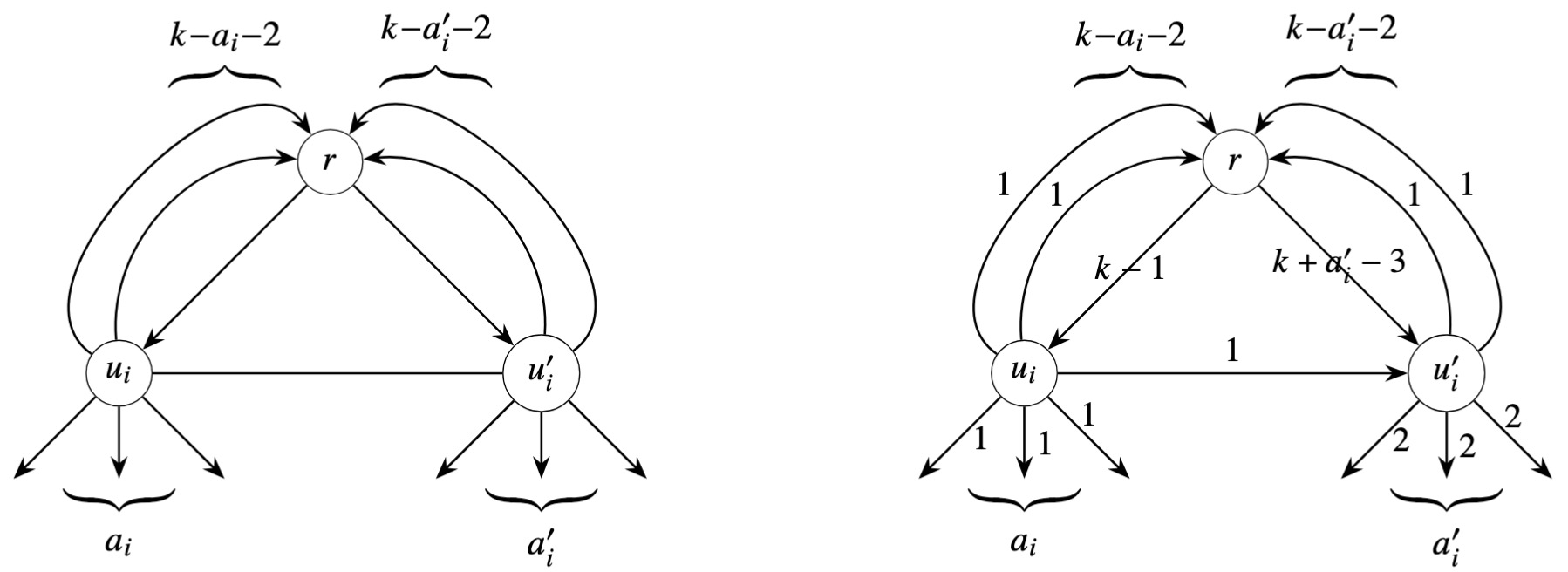}
    \caption{Left: part of graph $G=(V,E)$ and a partial orientation $\vec{F}$. Right: part of a nowhere-zero $k$-flow $(\vec{E},f)$ of $G$ such that $\vec{F}\subseteq \vec{E}$.}
    \label{fig:complexity-CBO-small}
\end{figure}

    We claim that $\vec{F}$ can be completed into a $k$-cut-balanced orientation if and only if the instance of SAT is satisfiable. We first prove the reverse direction. Given an assignment to the variables $x_1,...,x_n$ satisfying, complete $\vec{F}$ into an orientation $\vec{E}$ of all edges in the following way. Orient $(u_i,u'_i)\in \vec{E}$ if $x_i=0$;  orient $(u'_i,u_i)\in \vec{E}$ if $x_i=1$. We claim that $\vec{E}$ is a $k$-cut-balanced orientation. To prove this, we construct a nowhere-zero $k$-flow $f$ over $\vec{E}$. For an arbitrary $i\in [n]$, suppose $x_i=0$, which implies $(u_i,u'_i)\in\vec{E}$ (the case where $x_i=1$ is symmetric). Let $f(r,u_i)=k-1$ and $f(u_i,w)=1\ \forall (u_i,w)\in\delta_{\vec{E}}^+(u_i)$. Let $f(r,u'_i)=k+a'_i-3\in\{k-1,k-2\}$; $f(u'_i,w)=1\ \forall (u_i,w)\in\delta_{\vec{E}}^+(u_i)$ and $w=r$; and $f(u'_i,w)=2\ \forall (u_i,w)\in\delta_{\vec{E}}^+(u_i)$ and $w\neq r$ (see Figure \ref{fig:complexity-CBO-small} right). We note that $f$ satisfies flow conservation at both $u_i$ and $u'_i$. For each $j\in [m]$, since $C_j$ is satisfied, w.l.o.g. there is some $i\in[n]$ such that $x_i$ appears negatively in $C_j$ and $x_i=0$ (the case where there is some $i\in[n]$ such that $x_i$ appears positively in $C_j$ and $x_i=1$ is symmetric), which implies $f(u'_i,v_j)=2$. Therefore, $|C_j|+1\leq f(\delta_{\vec{E}}^-(v_j))\leq 2|C_j|$. Since $|\delta_{\vec{E}}^+(v_j)|=|C_j|+1$, we can arbitrarily assign flow values $1$ or $2$ to arcs in $\delta_{\vec{E}}^+(v_j)$ such that $f(\delta_{\vec{E}}^+(v_j))=f(\delta_{\vec{E}}^-(v_j))$. So far, we have assigned flow value for each arc in $\vec{E}$ and checked flow conservation for $u_i\ \forall i$ and $v_j\ \forall j$. Therefore, the flow conservation has to hold for $r$ as well. This shows that $(\vec{E},f)$ is a nowhere-zero $k$-flow. By Lemma \ref{lemma:equivalence}, $\vec{E}$ is a $k$-cut-balanced orientation.
\begin{figure}[htbp]
    \centering
    \includegraphics[width=\linewidth]{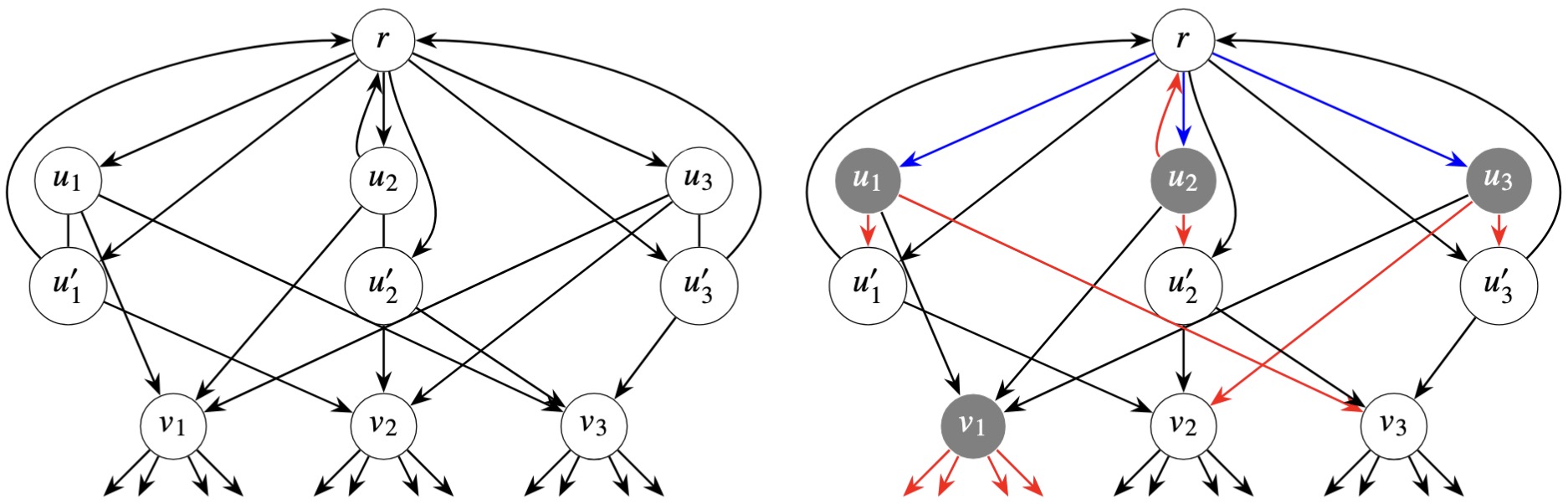}
    \caption{Left: the graph $G$ for restricted SAT instance $(x_1\vee x_2\vee x_3)\wedge (\bar{x}_1\vee \bar{x}_2\vee x_3)\wedge (x_1\vee \bar{x}_2\vee \bar{x}_3)$ and $k=4$. Right: an infeasible assignment $x_1=x_2=x_3=0$ yields a cut $X=\{u_1,u_2,u_3,v_1\}$ that violates the $k$-cut-balancedness condition. Here, arcs in $\delta^+(X)$ are colored red; arcs in $\delta^-(X)$ are colored blue.}
    \label{fig:complexity-CBO-big}
\end{figure}
    
    Next, we prove the forward direction. Let $\vec{E}$ be a $k$-cut-balanced orientation such that $\vec{F}\subseteq \vec{E}$. We assign $x_i=0$ if $(u_i,u'_i)\in \vec{E}$; $x_i=1$ if $(u'_i,u_i)\in \vec{E}$. We now show that all clauses are satisfied under this assignment. Suppose not, which means there is some clause $j\in [m]$ that is not satisfied. Consider the set $U:=\{u_i:x_i \text{ appears positively in } C_j\}$ and $U':=\{u'_i:x_i \text{ appears negatively in } C_j\}$. Since $C_j$ is unsatisfied, it follows that $(u_i,u'_i)\in \vec{E}\ \forall u_i\in U$ and $(u'_i,u_i)\in \vec{E}\ \forall u'_i\in U'$. Then, for every $u_i\in U$, $|\delta^+(u_i)|=(k-a_i-2)+a_i+1=k-1$, consisting of $k-a_i-2$ arcs to $r$, $a_i$ arcs to $v_j$'s and one arc to $u_i'$. Similarly,  for every $u'_i\in U'$, $|\delta^+(u'_i)|=k-1$. Let $X:=U\cup U'\cup \{v_j\}$. Thus, $|\delta^+(X)|=\sum_{i:u_i\in U}|\delta^+(u_i)|+\sum_{i:u'_i\in U'}|\delta^+(u'_i)|-|\delta^-(v_j)|+|\delta^+(v_j)|=(k-1)|C_j|-|C_j|+(|C_j|+1)=(k-1)|C_j|+1$. On the other hand, the only incoming arcs of $X$ are from the root $r$ to $U\cup U'$, and thus $|\delta^-(X)|=|U|+|U'|=|C_j|$. This violates the $k$-cut-balancedness condition $|\delta^+(X)|\leq \frac{k-1}{k}|\delta(X)|$, or equivalently $|\delta^+(X)|\leq (k-1)|\delta^-(X)|$, a contradiction.
\end{proof}

Theorem \ref{thm:hardness-WNZF-WCBO} follows as a corollary.
\begin{proof}[Proof of Theorem \ref{thm:hardness-WNZF-WCBO}]
    We first prove that it is NP-hard to approximate $\operatorname{WCBO}(k)$ within any finite factor for every finite $k
    \geq 3$. We reduce from the decision problem of whether $\vec{F}$ can be completed into a $k$-cut-balanced orientation. Let $\vec{G}=(V,E^+\cup E^-)$ be the bidirected graph of $G$. Set cost $c(e)=0$ and $c(e^{-1})=\infty$ for every arc $e\in \vec{F}$. Set cost $c(e^+)=c(e^-)=0$ for every $e\in E\setminus F$. A finite approximation algorithm of the minimum $k$-cut-balanced orientation returns a solution of cost $0$ if and only if $\vec{F}$ can be completed into a $k$-cut-balanced orientation.
    
    Next we prove that it is NP-hard to approximate  $\operatorname{WNZF}(k)$ within any finite factor for every finite $k
    \geq 3$. Letting $\beta=1$ in Lemma \ref{lemma:approx_translate}, we obtain that for every finite $k$, $\operatorname{WNZF}(k)$ has a finite factor approximation algorithm if and only if $\operatorname{WCBO}(k)$ has a finite factor approximation algorithm. Therefore, the hardness of approximating $\operatorname{WNZF}(k)$ within any finite factor follows from that of $\operatorname{WCBO}(k)$.
\end{proof}

\section{Bicriteria approximations for $\operatorname{WNZF}(k)$ and $\operatorname{WCBO}(k)$}\label{sec:algo-WNZF+WCBO}
In this section, we design bicriteria approximations for $\operatorname{WNZF}(k)$ and $\operatorname{WCBO}(k)$. Our algorithms are based on an LP-relaxation for both problems. 

We begin with a unified LP-relaxation for both problems. 
Let $G=(V, E)$ be the input graph, $\vec{G}=(V,E^+\cup E^-)$ be its bidirected graph, and $c: E^+ \cup E^-$ be the given cost. For each edge $e\in E$, we introduce indicator variables $y(e^+)$ and $y(e^-)$ to indicate the orientation to be used for the edge and non-negative integer variables $z(e^+)$ and $z(e^-)$ for the flow value on the two orientations of the edge. For the $y$ variables that indicate orientation, we impose the constraint that $y(e^+) + y(e^-)=1$ for each edge $e\in E$, since each edge is to be oriented in exactly one of the two directions. For the flow variables $z$, we impose the constraint that $z(\delta_{\vec{G}}^+(v))=z(\delta_{\vec{G}}^-(v))$ for each vertex $v\in V$ to conserve flow. We also impose the constraint that $y(e)\le z(e)\le (k-1)y(e)$ for every arc $e\in E^+\cup E^-$ since we want the flow variables $z$ to send flow only along the direction of the oriented arc and the capacity of the flow to be at most $k-1$. Together, we obtain the LP-relaxation given in \eqref{eq:k-cut-balanced+nwz-k-flow}. 
Thus, it follows from Lemma \ref{lemma:equivalence} that every integral feasible solution $y,z\in \Z^{E^+\cup E^-}$ to \eqref{eq:k-cut-balanced+nwz-k-flow} corresponds to a $k$-cut-balanced orientation and a nowhere-zero $k$-flow, respectively. 
For $\operatorname{WCBO}(k)$, we let $c_y=c$ and $c_z=0$. For $\operatorname{WNZF}(k)$, we let $c_y=0$ and $c_z=c$.

\begin{equation}\label{eq:k-cut-balanced+nwz-k-flow}
    \begin{aligned}
    \min~&c_y^\intercal y+c_z^\intercal z\\
    s.t.~&z(\delta_{\vec{G}}^+(v))=z(\delta_{\vec{G}}^-(v))\quad \forall v\in V\\
    &y(e^+)+y(e^-)=1\quad \forall e\in E\\
    &y\leq z\leq (k-1)y\\
        &y\geq 0.
    \end{aligned}
\end{equation}

In Sections \ref{sec:WNZF-LP} and \ref{sec:WNZF-approx}, we focus on $\operatorname{WNZF}(k)$: we consider the projection of \eqref{eq:k-cut-balanced+nwz-k-flow} on the $z$ variables and prove certain extreme point properties of the projected LP in Section \ref{sec:WNZF-LP}. We use these properties to design a bicriteria approximation for $\operatorname{WNZF}(k)$ in Section \ref{sec:WNZF-approx} thereby proving Theorem \ref{thm:approx-nwz-kflow}.
%We recall that an $(\alpha,\beta)$-approximation algorithm for $\operatorname{WCBO}(k)$ returns a $\beta k$-cut-balanced orientation with cost at most $\alpha c(\opt)$ where $c(\opt)$ is the minimum cost of a $k$-cut-balanced orientation. It follows from 
Lemma \ref{lemma:approx_translate} and Theorem \ref{thm:approx-nwz-kflow} imply a $(6(k-1),6)$-approximation algorithm for $\operatorname{WCBO}(k)$. Next, we improve on this bicriteria approximation to achieve a $(k, 6)$-approximation for $\operatorname{WCBO}(k)$.  
In Sections \ref{sec:WCBO-LP} and \ref{sec:WCBO-approx}, we focus on $\operatorname{WCBO}(k)$: we consider the projection of \eqref{eq:k-cut-balanced+nwz-k-flow} on the $y$ variables and prove certain extreme point properties of the projected LP in Section \ref{sec:WCBO-LP}. We use these properties to design a bicriteria approximation for $\operatorname{WCBO}(k)$ in Section \ref{sec:WCBO-approx}, thereby proving Theorem \ref{thm:approx-cut-balanced}.

\subsection{LP relaxation for $\operatorname{WNZF}(k)$ and its extreme point structure}\label{sec:WNZF-LP}
%In this section, we give a natural LP relaxation for both $\operatorname{WNZF}(k)$ and $\operatorname{WCBO}(k)$. 
We can project the feasible region of \eqref{eq:k-cut-balanced+nwz-k-flow} on the $z$ variables, which yields an LP relaxation for $\operatorname{WNZF}(k)$.
\begin{equation}\label{eq:nwz-k-flow}
    \begin{aligned}
    \min~&c^\intercal z\\
    s.t.~&z(\delta_{\vec{G}}^+(v))=z(\delta_{\vec{G}}^-(v))\quad \forall v\in V\\
    &1\leq z(e^+)+z(e^-)\leq k-1\quad \forall e\in E\\
        &z\geq 0.
    \end{aligned}\tag{$\mathcal{P}_k$}
\end{equation}

\begin{lemma}\label{lemma:proj-z}
    The polytope in \eqref{eq:nwz-k-flow} is the projection of the polytope in \eqref{eq:k-cut-balanced+nwz-k-flow} to the $z$ variables.
\end{lemma}
% We prove this lemma in Appendix \ref{sec:proj-LP}.
\begin{proof}
    We note that for every $(y,z)$ feasible for \eqref{eq:k-cut-balanced+nwz-k-flow}, $z$ is feasible for \eqref{eq:nwz-k-flow}. Conversely, let $z$ be feasible for \eqref{eq:nwz-k-flow}. We define $y(e^+):=\frac{z(e^+)}{z(e^+)+z(e^-)}$ and $y(e^-):=\frac{z(e^-)}{z(e^+)+z(e^-)}\ \forall e\in E$. It holds that
    \[
    \frac{1}{k-1}z(e^+)\leq \frac{z(e^+)}{z(e^+)+z(e^-)}\leq z(e^+).
    \]
    Thus, $\frac{1}{k-1}z(e^+)\leq y(e^+)\leq z(e^+)$. The same inequality holds for $e^-$. Moreover, $y(e^+)+y(e^-)=1$. Thus, $(y,z)$ is feasible for \eqref{eq:k-cut-balanced+nwz-k-flow}.
\end{proof}

\begin{remark}
If we drop the upper bound constraint $z(e^+)+z(e^-)\leq k-1$ for every $e\in E$ from \eqref{eq:nwz-k-flow}, then the resulting LP is an LP-relaxation of the asymmetric postman problem and it has been shown to be half-integral in several earlier works \cite{kappauf1979mixed, win1989windy, ralphs1993mixed, zaragoza2003postman}. 
%. Kappauf and Koehler \cite{kappauf1979mixed},  Win \cite{win1989windy}, Ralphs \cite{ralphs1993mixed} and Zaragoza Martinez \cite{zaragoza2003postman} independently proved that the resulting LP is half-integral.
\end{remark}

We show that \eqref{eq:nwz-k-flow} is half-integral. Our proof is inspired by extreme point results for postman problems on mixed graphs \cite{zaragoza2003postman}. In fact, we show additional stronger properties as stated in the following theorem.
% We defer its proof to Appendix \ref{sec:LP-half-integral} due to space constraints.

\begin{theorem}\label{thm:APP_half_integral}
Let $z^*$ be an extreme point optimal solution to \eqref{eq:nwz-k-flow}. Then,
% \begin{enumerate}[label={(\arabic*)}]
%     \item for every $e\in E$, if $e^+,e^-\in \supp(z^*)$, then $z^*(e^+)+z^*(e^-)=1$, 
%     \item for every $e\in E$, $z^*(e^+)\in \Z$ if and only if $z^*(e^-)\in \Z$, 
%     \item for every $e\in E^+\cup E^-$, if $z^*(e)$ is non-integral, then $z^*(e)=\frac{1}{2}$, and 
%     \item the integral arcs of $z^*$ form a $k$-flow.
% \end{enumerate}
\begin{enumerate}[label={(\arabic*)}]
    \item for every $e\in E$, if $e^+,e^-\in \supp(z^*)$, then $z^*(e^+)+z^*(e^-)=1$, 
    \item for every $e\in E$, $z^*(e^+)\in \Z$ if and only if $z^*(e^-)\in \Z$, 
    \item for every $e\in E^+\cup E^-$, if $z^*(e)$ is non-integral, then $z^*(e)=\frac{1}{2}$, and 
    \item the integral arcs of $z^*$ form a $k$-flow.
\end{enumerate}
\end{theorem}
\begin{proof}
        We first prove (1). Suppose $e^+,e^-\in \supp(z^*)$ and $z^*(e^+)+z^*(e^-)>1$. Decreasing both $z^*(e^+)$ and $z^*(e^-)$ by some sufficiently small $\epsilon>0$ gives a feasible solution whose cost is at most $c^\intercal z^*$, 
        % \knote{whose cost is less than that of $z^*$}
        contradicting the fact that $z^*$ is an extreme point optimal solution.
        
        We show an intermediate claim that will help prove the remaining three statements of the theorem. Let $F:=\{e\in E: z^*(e^+)\notin \Z \text{ or } z^*(e^-)\notin \Z\}$, which we call the \emph{fractional edges} w.r.t. $z^*$. We claim that $F$ is acyclic. Suppose not. Let $C\subseteq F$ be an undirected cycle. Let $\vec{C}$ be the corresponding directed cycle, where all edges are oriented clockwise. Let $\cev{C}:=(\vec{C})^{-1}$ be corresponding directed cycle oriented counterclockwise. It follows from (1) that for every $e\in F$, if $z(e^+)\notin \Z$, then either $z^*(e^-)\notin \Z$ or $z^*(e^-)=0$. Therefore, there are three types of arcs $e\in \vec{C}\cup \cev{C}$ with $z^*(e)>0$: $C^0:=\{e\in \vec{C}\cup \cev{C}: z^*(e),z^*(e^{-1})\notin \Z\}$, $C^+:=\{e\in \vec{C}: z^*(e)\notin \Z, z^*(e^{-1})=0\}$, and $C^-:=\{e\in \cev{C}: z^*(e)\notin \Z, z^*(e^{-1})=0\}$. For a sufficiently small $\epsilon>0$, let
        \[
        \begin{aligned}
            z'(e)=\begin{cases}
            z^*(e)+\epsilon, & e\in C^+\\
            z^*(e)-\epsilon, & e\in C^-\\
            z^*(e)+\epsilon/2, & e\in C^0\cap \vec{C}\\
            z^*(e)-\epsilon/2, & e\in C^0\cap \cev{C}\\
            z^*(e), & o/w.
        \end{cases}
        \end{aligned}
        \quad\text{and}\quad
        \begin{aligned}
            z''(e)=\begin{cases}
            z^*(e)-\epsilon, & e\in C^+\\
            z^*(e)+\epsilon, & e\in C^-\\
            z^*(e)-\epsilon/2, & e\in C^0\cap \vec{C}\\
            z^*(e)+\epsilon/2, & e\in C^0\cap \cev{C}\\
            z^*(e), & o/w.
        \end{cases}
        \end{aligned}
        \]
We claim that both $z'$ and $z''$ are feasible to \eqref{eq:nwz-k-flow}: we prove this for $z'$ and the statement for $z''$ follows by symmetric arguments. Flow conservation holds because we obtained $z'$ from $z^*$ by pushing $\epsilon$ units of flow along $\vec{C}$ (it is equivalent to view pushing $-\epsilon$ flow along $e\in \cev{C}$ as pushing $\epsilon$ flow along $e\in \vec{C}$). Moreover, for $e\in C^+\cup C^-$, since $z^*(e)\notin \Z$ and $z^*(e^{-1})=0$, we have that $1<z^*(e)+z^*(e^{-1})<k-1$. Thus, for sufficiently small $\epsilon>0$, $1\leq z'(e)+z'(e^{-1})\leq k-1$. Finally, $z'(e)\geq 0$ follows from the fact that for every $e\in C^+\cup C^-\cup C^0$, $z^*(e)>0$. This proves that both $z'$ and, by symmetry, $z''$ are feasible to \eqref{eq:nwz-k-flow}. However, $z^*=(z'+z'')/2$, which is a contradiction to the fact that $z^*$ is an extreme point. 

We next prove (2) and (3) together by induction. Since $F$ is acyclic, it forms a forest. We remove the leaves of $F$ one by one and delete non-integral edges along the way, maintaining the property that the flow $z^*$ is integral on $E\setminus F$. Pick a leaf node $v$ of $F$ incident to a unique edge $e=(u,v)\in F$. If $z^*(e^+)\in \Z$, we immediately obtain that $z^*(e^-)\in \Z$ due to the flow conservation at $v$ and the fact that every arc $f$ incident to $v$ other than $e^-$ has integral flow $z^*(f)$. Therefore, we may assume that both $z^*(e^+), z^*(e^-)\notin \Z$. In this case, $z^*(e^+)-z^*(e^-)\in \Z$ by flow conservation at $v$. It follows from (1) that $z^*(e^+)+z^*(e^-)=1$. Therefore, we conclude that $z^*(e^+)=z^*(e^-)=1/2$. We delete $e^+$ and $e^-$ together. Since $z^*(e^+)=z^*(e^-)$, the property that $z^*$ is a circulation is preserved. This proves that (2) $z^*(e^+)\in \Z$ if and only if $z^*(e^-)\in \Z$; (3) if $z^*(e)$ is non-integral, then $z^*(e)=\frac{1}{2}$. 

We finally prove (4). It follows from (3) that the integral arcs of $z^*$ form a circulation. Moreover, $z^*(e)\leq k-1\ \forall e\in E^+\cup E^-$. Furthermore, if both $z^*(e^+),z^*(e^-)\in \Z$, it follows from (1) that exactly one of them is in $\supp(z^*)$. Thus, $\supp(z^*)$ is a partial orientation which induces a $k$-flow.
\end{proof}

\subsection{Bicriteria approximation for $\operatorname{WNZF}(k)$}\label{sec:WNZF-approx}
In this section, we prove Theorem \ref{thm:approx-nwz-kflow} by giving a bicriteria approximation algorithm for $\operatorname{WNZF}(k)$, for every finite integer $k\geq 6$ and for $k=\infty$. We recall that an algorithm is an $(\alpha,\beta)$-approximation for $\operatorname{WNZF}(k)$ if we return a nowhere-zero $\beta k$-flow with cost at most $\alpha c(\opt)$ where $c(\opt)$ is the minimum cost of a nowhere-zero $k$-flow. 
%Our algorithm proceeds as follows: Solve the LP relaxation \eqref{eq:nwz-k-flow} and let $z^*$ be an extreme point optimal solution. According to Theorem \ref{thm:APP_half_integral} (4), the integral arcs of $z^*$ form a $k$-flow, denoted as $f$. Let $g$ be an arbitrary nowhere-zero $6$-flow of $E$, which can be computed in polynomial time \cite{younger1983integer}. Then, $-g$ is also a nowhere-zero $6$-flow. Our algorithm returns $6f+g$ or $6f-g$, whichever has a smaller cost. Now, we prove its performance guarantee as stated in Theorem \ref{thm:approx-nwz-kflow}.

\begin{proof}[Proof of Theorem \ref{thm:approx-nwz-kflow}]
Our algorithm proceeds as follows: Solve the LP relaxation \eqref{eq:nwz-k-flow} and let $z^*$ be an extreme point optimal solution. According to Theorem \ref{thm:APP_half_integral} (4), the integral arcs of $z^*$ form a $k$-flow, denoted as $f$. Let $g$ be an arbitrary nowhere-zero $6$-flow of $E$, which can be computed in polynomial time \cite{younger1983integer}. Then, $-g$ is also a nowhere-zero $6$-flow. Our algorithm returns $6f+g$ or $6f-g$, whichever has a smaller cost. Now, we bound its approximation guarantee. 

    According to Proposition \ref{prop:sum-flows}, both $6f+g$ and $6f-g$ are nowhere-zero $6k$-flows. We prove that $\min\{c(6f+g),c(6f-g)\}\leq 6 c(z^*)$. Let $E_1:=\supp(f)$ be the edges oriented by $f$, i.e., $\{e\in E^+\cup E^-: f(e)\neq 0\}$. Let $\vec{E}_1:=\supp^+(f)$ be the orientation associated with $f$, i.e., $\{e\in E^+\cup E^-: f(e)>0\}$. Let $E_2=E\setminus E_1$ and $\vec{E}_2$ be the orientation associated with $g$ restricted to $E_2$. Then, $\cev{E}_2=(\vec{E}_2)^{-1}$ is the orientation associated with $-g$ restricted to $E_2$. Let $\vec{E}:=\vec{E}_1\cup \vec{E}_2$, which will be the orientation associated with $6f+g$ according to Proposition \ref{prop:sum-flows}. Similarly, let $\vec{E}':=\vec{E}_1\cup \cev{E}_2$, which will be the orientation associated with $6f-g$. Then,
    
\[
\begin{aligned}
    &\min\{c(6f+g),c(6f-g)\}\\
    &\quad \quad \leq \frac{1}{2}\Big(\sum_{e\in \vec{E}} c(e)(6f(e)+g(e))+\sum_{e\in \vec{E}'} c(e)(6f(e)-g(e))\Big)\\
    &\quad \quad =\frac{1}{2}\sum_{e\in \vec{E}_1} c(e)\Big(\big(6f(e)+g(e)\big)+\big(6f(e)-g(e)\big)\Big)+\frac{1}{2}\sum_{e\in \vec{E}_2} c(e)g(e)+\frac{1}{2}\sum_{e\in \cev{E}_2} c(e)(-g)(e)\\
    &\quad \quad =6\sum_{e\in \vec{E}_1} c(e)f(e)+\frac{1}{2}\sum_{e\in \vec{E}_2} \Big(c(e)g(e)+c(e^{-1})g(e)\Big)\\
    &\quad \quad \leq 6\sum_{e\in \vec{E}_1} c(e)f(e)+\frac{1}{2}\sum_{e\in \vec{E}_2} 5\Big(c(e)+c(e^{-1})\Big)\\
    &\quad \quad =6\sum_{e\in \vec{E}_1} c(e)z^*(e)+\sum_{e\in \vec{E}_2} 5\Big(c(e)z^*(e)+c(e^{-1})z^*(e^{-1})\Big)\\
    &\quad \quad \leq 6c^\top z^*,
\end{aligned}
\]
where the second equality follows from the definition of $-g$, in which $(-g)(e)=-g(e)=g(e^{-1})$. The second inequality follows from the fact that $g$ is a nowhere-zero $6$-flow and thus $g(e)\leq 5\ \forall e\in \vec{E}_2$. The third equality follows from the fact that $f(e)=z^*(e)\ \forall e\in \vec{E}_1$ and that $z^*(e)=z^*(e^{-1})=\frac{1}{2}\ \forall e\in E_2$ according to Theorem \ref{thm:APP_half_integral} (3). 
    
    % \[
    % \sum_{e\in \vec{E}}c(e)f(e)=\sum_{e\in \vec{E_1}} c(e)f(e)= \sum_{e\in \vec{E_1}} c(e)z^*(e).
    % \]
    % \[
    % \sum_{e\in \vec{E}}c(e)g(e)\leq \sum_{e\in \vec{E}}5c(e)\leq \sum_{e\in \vec{E_1}}5c(e)z^*_e+\sum_{e\in \vec{E_2}}10c(e)z^*_e,
    % \]
    % where the first inequality follows from the fact that $g(e)\leq 5\ \forall e\in \vec{E}$. The second inequality follows from the fact that $z^*_e\geq 1\ \forall e\in \vec{E}_1$ and $z^*_e=1/2\ \forall e\in \vec{E}_2$. Combining them, we get
    % \[
    % c(6f+g)=6\sum_{e\in \vec{E}}c(e)f(e)+\sum_{e\in \vec{E}}c(e)g(e)\leq 6\sum_{e\in \vec{E_1}} c(e)z^*(e)+\sum_{e\in \vec{E_1}}5c(e)z^*_e+\sum_{e\in \vec{E_2}}10c(e)z^*_e\leq 11 \sum_{e\in \vec{E}}c(e)z^*_e.
    % \]
\end{proof}

%By taking $k=\infty$, we get a $6$-approximation algorithm for $\operatorname{WNZF}(\infty)$, the min cost nowhere-zero flow problem, and thus Corollary \ref{cor:approx-nwz-flow} follows. 
% Note that if we do not impose that at most one of $e^+$ and $e^-$ has positive flow, this is precisely the asymmetric postman problem, which has a $\frac{3}{2}$-approximation algorithm.

%\subsection{Min cost $k$-cut-balanced orientations}
\subsection{LP relaxation for $\operatorname{WCBO}(k)$ and its extreme point structure}\label{sec:WCBO-LP}
We can project the feasible region of \eqref{eq:k-cut-balanced+nwz-k-flow} on the $y$ variables, which yields an LP relaxation for $\operatorname{WCBO}(k)$.
\begin{equation}\label{eq:k-cut_balanced}
\begin{aligned}
    \min~&c^\intercal y\\
    s.t.~&y(\delta_{\vec{G}}^+(U))\leq \frac{k-1}{k}|\delta_E(U)|\quad \forall U\subseteq V\\
    &y(e^+)+y(e^-)=1\quad \forall e\in E\\
        &y\geq 0.
\end{aligned}\tag{$\mathcal{Q}_k$}
\end{equation}

\begin{lemma}\label{lemma:proj-y}
    The polytope in \eqref{eq:k-cut_balanced} is the projection of the polytope in \eqref{eq:k-cut-balanced+nwz-k-flow} on the $y$ variables.
\end{lemma}
% We prove this lemma in Appendix \ref{sec:proj-LP}.
\begin{proof}
    Let $(y,z)$ be a feasible solution to \eqref{eq:k-cut-balanced+nwz-k-flow}. Then, for every $U\subseteq V$,
    \[
    \begin{aligned}
        y(\delta_{\vec{G}}^+(U))\leq z(\delta_{\vec{G}}^+(U))=&z(\delta_{\vec{G}}^-(U))\leq (k-1)y(\delta_{\vec{G}}^-(U)),\\
        y(\delta_{\vec{G}}^+(U))+y(\delta_{\vec{G}}^-(U))=&\sum_{e\in \delta_E(U)}\big(y(e^+)+y(e^-)\big)=|\delta_E(U)|.
    \end{aligned}
    \]
    Therefore, \[|\delta_E(U)|=y(\delta_{\vec{G}}^+(U))+y(\delta_{\vec{G}}^-(U))\geq y(\delta_{\vec{G}}^+(U))+\frac{1}{k-1}y(\delta_{\vec{G}}^+(U))=\frac{k}{k-1}y(\delta_{\vec{G}}^+(U)).\]
    Thus, $y$ is feasible for $\eqref{eq:k-cut_balanced}$.

    Let $y$ be a feasible solution to $\eqref{eq:k-cut_balanced}$. Suppose the least common multiple of the coordinates of $y$ is $M$. Construct a graph $G'=(V,E')$ with $M$ copies of each edge $e\in E$ and an orientation $\vec{E}'$ such that $My(e^+)$ arcs are oriented the same as $e^+$ and the remaining $M-My(e^+)=My(e^-)$ arcs are oriented the same as $e^-$, denoted by $\vec{E}'(e^+)$ and $\vec{E}'(e^-)$, respectively. Since every arc is multiplied by a same factor from the fractional $y$, the new graph satisfies
    \[
    |\delta_{\vec{E}'}^+(U)|\leq \frac{k-1}{k}|\delta_{E'}(U)|\quad \forall U\subseteq V,
    \]
    which means $\vec{E}'$ is a $k$-cut-balanced orientation of $G'$. By Lemma \ref{lemma:equivalence}, $\vec{E}'$ induces a nowhere-zero $k$-flow $z':\vec{E}'\rightarrow \{1,2,...,k-1\}$. Let $z(e^+):=\frac{1}{M}\sum_{e'\in \vec{E}'(e^+)} z'(e')$ and $z(e^-):=\frac{1}{M}\sum_{e'\in \vec{E}'(e^-)} z'(e')\ \forall e\in E$. It follows from $z'$ being a flow that $z$ is a flow. Moreover, for every $e\in E^+\cup E^-$, $y(e)=\frac{1}{M} |\vec{E}'(e)|\leq z(e)=\frac{1}{M}\sum_{e'\in \vec{E}'(e)} z'(e')\leq \frac{1}{M}(k-1) |\vec{E}'(e)|=(k-1)y(e)$. Thus, $(y,z)$ is feasible for \eqref{eq:k-cut-balanced+nwz-k-flow}.
\end{proof}
To prove extreme point properties of \eqref{eq:k-cut_balanced}, we need the following seminal theorem of Edmonds and Giles \cite{edmonds1977min}.
\begin{theorem}[\cite{edmonds1977min}]\label{thm:submodular_flow}
    Let $D=(V,A)$ be a digraph and $f:2^V\rightarrow \R$ be a submodular function. Consider the \emph{submodular flow} polyhedron defined as follows:
    \[
P(f):=\big\{y: A\rightarrow \R \mid y(\delta^+(U))-y(\delta^-(U))\leq f(U)\quad \forall U\subseteq V\big\}.
\]
If $f$ is integral, then $P(f)$ is  \emph{box-integral}, i.e., $P(f)\cap \{x: l\leq x\leq u\}$ is integral for every $l,u\in \Z$. 
\end{theorem}
The following theorem proves that after projecting \eqref{eq:k-cut_balanced} onto $\{y(e^+): e^+\in E^+\}$, we obtain a submodular flow where $f$ is $1/k$-integral.
\begin{lemma}\label{lemma:1/k-integral}
    The extreme points of \eqref{eq:k-cut_balanced} are $1/k$-integral.
\end{lemma}
\begin{proof}
    For $y\in \R^{E^+\cup E^-}$, denote by $y|_{E^+}$ the restriction of $y$ to $E^+$. We claim that the projection of \eqref{eq:k-cut_balanced} onto $y|_{E^+}$ is
    \begin{equation}\label{eq:CBO_subflow}
\begin{aligned}
    P:=\Big\{y\in[0,1]^{E^+}:y(\delta_{E^+}^+(U))-y(\delta_{E^+}^-(U))\leq \frac{1}{k}\Big((k-1)|\delta_{E^+}^+(U)|-|\delta_{E^+}^-(U)|\Big)\quad \forall U\subseteq V\Big\}.
\end{aligned}
\end{equation}
To see this, let $y$ be a feasible solution for \eqref{eq:k-cut_balanced},  
% For every $e^-\in E^-$, replace $y(e^-)$ with $1-y(e^+)$. 
Then, for $U\subseteq V$, we have that
    \[
    \begin{aligned}
        y(\delta^+_{\vec{G}}(U))=&y(\delta^+_{E^+}(U))+y(\delta^+_{E^-}(U))=y(\delta^+_{E^+}(U))+\sum_{e^-\in \delta^+_{E^-}(U)} y(e^-)\\
        =&y(\delta^+_{E^+}(U))+\sum_{e^+\in \delta^-_{E^+}(U)} (1-y(e^+))
        =y(\delta^+_{E^+}(U))+|\delta^-_{E^+}(U)|-y(\delta^-_{E^+}(U)),
    \end{aligned}
    \]
    where the third equality follows from $y(e^+)+y(e^-)=1\ \forall e\in E$. Thus, it follows from $y(\delta_{\vec{G}}^+(U))\leq \frac{k-1}{k}|\delta_E(U)|$ that
    \[
    \begin{aligned}
        &y(\delta^+_{E^+}(U))-y(\delta^-_{E^+}(U))\leq \frac{k-1}{k}|\delta_E(U)|-|\delta^-_{E^+}(U)|\\
        =&\frac{k-1}{k}\Big(|\delta_{E^+}^+(U)|+|\delta_{E^+}^-(U)|\Big)-|\delta^-_{E^+}(U)|
        =\frac{1}{k}\Big((k-1)|\delta_{E^+}^+(U)|-|\delta_{E^+}^-(U)|\Big).
    \end{aligned}
    \]
This implies that $y|_{E^+}\in P$, where $P$ is defined as \eqref{eq:CBO_subflow}. On the other hand, for every $y\in P$, define $\tilde{y}\in \R^{E^+\cup E^-}$ as $\tilde{y}(e^+):=y(e^+)$ and $\tilde{y}(e^-):=1-y(e^+)\ \forall e\in E$. The same argument shows that $\tilde{y}$ is feasible to \eqref{eq:k-cut_balanced}. Therefore, $P$ is the projection of \eqref{eq:k-cut_balanced} onto $y|_{E^+}$.
    
Let $f:2^V\rightarrow \R$ be defined as for $U\subseteq V$,
\[f(U):=\frac{1}{k}\Big((k-1)|\delta_{E^+}^+(U)|-|\delta_{E^+}^-(U)|\Big)=\frac{1}{k}\Big(\big(|\delta_{E^+}^+(U)|-|\delta_{E^+}^-(U)|\big)+(k-2)|\delta_{E^+}^+(U)|\Big).
\]
Since $|\delta_{E^+}^+(U)|-|\delta_{E^+}^-(U)|=\sum_{v\in U}(|\delta_{E^+}^+(v)|-|\delta_{E^+}^-(v)|)$ is modular, $|\delta_{E^+}^+(U)|$ is submodular, and $k\geq 2$, we conclude that $f$ is a $1/k$-integral submodular function. By Theorem \ref{thm:submodular_flow}, since $k\cdot f$ is integral, $P$ is $1/k$-integral. Thus, for every $y\in P$, $y=\sum_{i=1}^{t}\lambda_i y_i$ for some $\sum_{i=1}^t\lambda_i=1,\ \lambda\geq 0$ and $1/k$-integral $y_1,...,y_t\in P$. This implies that for every $\tilde{y}=(y,1-y)\in\R^{E^+}\times \R^{E^-}$ feasible for \eqref{eq:k-cut_balanced},
\[\tilde{y}=(y,1-y)=\Big(\sum_{i=1}^{t}\lambda_i y_i,1-\sum_{i=1}^{t}\lambda_i y_i\Big)=\sum_{i=1}^{t}\lambda_i (y_i, 1-y_i)=\sum_{i=1}^{t}\lambda_i\tilde{y}_i,
\]
where $\tilde{y}_i$ is $1/k$-integral and feasible for \eqref{eq:k-cut_balanced}. This implies that the extreme points of \eqref{eq:k-cut_balanced} are $1/k$-integral.
\end{proof}

\subsection{Bicriteria approximation algorithms for $\operatorname{WCBO}(k)$}\label{sec:WCBO-approx}
Before going into the algorithm, we need one more definition. Given a graph $G=(V,E)$, a \emph{partial $k$-cut-balanced orientation} is an orientation $\vec{F}$ of a subset $F\subseteq E$ of edges such that $|\delta_{\vec{F}}^+(U)|\leq \frac{k-1}{k}|\delta_{F}(U)|\ \forall U\subseteq V$. We emphasize that the number of outgoing arcs is upper bounded by a fraction of the number of arcs in $F$ instead of $E$.

In the following lemma, we give a characterization for the existence of a partial $k$-cut-balanced orientation that extends a partial orientation. The same characterization for $k=2$ was given by Ford and Fulkerson \cite{ford2015flows} (see also \cite{kappauf1979mixed}).
\begin{lemma}\label{lemma:complete_k-balanced}
     Given a graph $G=(V,E)$, suppose a subset $E_1\subseteq E$ of edges  has been oriented as $\vec{E}_1$. Then, $\vec{E}_1$ can be extended to a partial $k$-cut-balanced orientation if and only if 
\begin{equation}\label{eq:partial_k_cut_balanced}
         |\delta_{\vec{E_1}}^+(U)|\leq \frac{k-1}{k}|\delta_E(U)|\quad \forall U\subseteq V.
     \end{equation}
     Moreover, if $\vec{E}_1$ satisfies \eqref{eq:partial_k_cut_balanced}, then such a partial $k$-cut-balanced orientation can be found in polynomial time.
\end{lemma}
\begin{proof}
    ``Only if" follows from the definition of partial $k$-cut-balanced orientations and the fact that $E_1\subseteq E$. We prove the ``if" direction. Bidirect the edges in $E_2:=E\setminus E_1$, denoted as $E_2^+\cup E_2^-$. In order to complete $\vec{E}_1$ into a directed Eulerian subgraph of $(V,\vec{E}_1\cup E_2^+\cup E_2^-)$, we need a flow induced on $\vec{E}_1\cup E_2^+\cup E_2^-$, which has capacity lower bound $1$ and upper bound $k-1$ for $e\in \vec{E}_1$, together with capacity lower bound $0$ and upper bound $k-1$ for $e\in  E_2^+\cup E_2^-$. By Hoffman's circulation theorem \cite{hoffman2003some}, there exists such a flow if and only if 

\[
    |\delta_{\vec{E}_1}^+(U)|\leq (k-1)|\delta_{\vec{E}_1}^-(U)|+(k-1)\Big(|\delta_{E}(U)|-|\delta_{\vec{E}_1}^-(U)|-|\delta_{\vec{E}_1}^+(U)|\Big)\quad \forall U\subseteq V,
\]
i.e.,
\[
    |\delta_{\vec{E}_1}^+(U)|\leq \frac{k-1}{k} |\delta_{E}(U)|\quad \forall U\subseteq V.
\]
Thus, by the given condition and Hoffman's circulation theorem, there exists a circulation $f$ satisfying $0\leq f(e)\leq k-1\ \forall e\in E^+\cup E^-$, and $\vec{E}_1\subseteq\supp(f)$.
Finally, if there is an edge $e\in E_2$ such that both $e^+$ and $e^-$ have positive flow, then we decrease them by the same amount so that at least one of them has flow value zero. Thus, we may assume that for every edge $e\in E$, at most one of $f(e^+),f(e^-)$ is nonzero. Therefore, $f$ is a $k$-flow. Let $F:=\supp(f)$ be the edges oriented by $f$ and $\vec{F}:=\supp^+(f)$ be the partial orientation associated with $f$. Applying Lemma \ref{lemma:equivalence} to the graph $(V,F)$, we conclude that $\vec{F}$ is a partial $k$-cut-balanced orientation. Moreover, we can use a polynomial-time circulation algorithm to construct $f$ (see e.g. \cite{ford2015flows,chen2022maximum}), and thus $\vec{F}$ can be constructed in polynomial time.
\end{proof}

\begin{remark}
The function $f:2^V\rightarrow \R$ defined as, for $U\subseteq V$,
    \[
         f(U):=\frac{k-1}{k}|\delta_E(U)|-|\delta_{\vec{E_1}}^+(U)|=\frac{1}{k}\Big(|\delta_{\vec{E}_1}^-(U)|-|\delta_{\vec{E}_1}^+(U)|\Big)+\frac{k-2}{k}|\delta_{\vec{E}_1}^-(U)|+\frac{k-1}{k}|\delta_{E\setminus E_1}(U)|
     \]
     is submodular, since $|\delta_{\vec{E}_1}^-(U)|-|\delta_{\vec{E}_1}^+(U)|$ is modular, $|\delta_{\vec{E}_1}^-(U)|$ and $|\delta_{E\setminus E_1}(U)|$ are submodular, and $k\geq 2$. By Lemma \ref{lemma:complete_k-balanced}, checking whether $\vec{E}_1$ can be extended to a partial $k$-cut-balanced orientation is equivalent to checking whether $f(U)\geq 0\ \forall U\subseteq V$, which can be done in polynomial time using submodular minimization~(see e.g. \cite{schrijver2000combinatorial,grotschel1981ellipsoid,iwata2001combinatorial}). In contrast, according to Theorem \ref{thm:complete-CBO-hard}, checking whether $\vec{E}_1$ can be extended to a (complete) $k$-cut-balanced orientation is NP-hard.
\end{remark}

Now we are ready to prove Theorem \ref{thm:approx-cut-balanced}. 
\begin{proof}[Proof of Theorem \ref{thm:approx-cut-balanced}]
Our algorithm proceeds as follows:  Solve the LP relaxation \eqref{eq:k-cut_balanced} and let $y^*$ be an extreme point optimal solution. Let $\vec{E}_1:=\{e\in E^+\cup E^-: y^*(e)=1\}$ and $E_1$ be its undirected counterpart. It follows that 
\[|\delta_{\vec{E_1}}^+(U)|\leq y^*(\delta_{\vec{E}}^+(U))\leq \frac{k-1}{k}|\delta_E(U)|\quad \forall U\subseteq V.\]
By Lemma \ref{lemma:complete_k-balanced}, $\vec{E}_1$ can be extended to a partial $k$-cut-balanced orientation $\vec{F}$ in polynomial time. Let $\vec{E}_2:=\vec{F}\setminus \vec{E}_1$ and $E_2$ be its undirected counterpart. Since $\vec{F}$ is $k$-cut-balanced, it induces a partial $k$-flow $f$ by Lemma \ref{lemma:equivalence}. Let $g$ be an arbitrary nowhere-zero $6$-flow of $(V,E)$, which can be constructed in polynomial time \cite{younger1983integer}. Our algorithm returns the orientation $\vec{E}$ associated with $6f+g$. Now, we bound its performance guarantee as stated in the theorem. 
%of Theorem \ref{thm:approx-cut-balanced}.

    According to Proposition \ref{prop:sum-flows}, $6f+g$ is a nowhere-zero $6k$-flow. Therefore, $\vec{E}$ is a $6k$-cut-balanced orientation.
    Next, we prove that $c(\vec{E})\leq k c(y^*)$. Let $E_3=E\setminus (E_1\cup E_2)$ and $\vec{E}_3$ be the orientation associated with $g$ restricted to $E_3$. According to Proposition \ref{prop:sum-flows}, $\vec{E}=\vec{E}_1\cup \vec{E}_2\cup\vec{E}_3$. Therefore, 
    \[
c(\vec{E})=\sum_{e\in \vec{E}} c(e)=\sum_{e\in \vec{E}_1} c(e)+\sum_{e\in \vec{E}_2\cup \vec{E}_3} c(e)\leq \sum_{e\in \vec{E}_1} c(e)y^*(e)+\sum_{e\in \vec{E}_2\cup \vec{E}_3} c(e)\cdot ky^*(e)\leq kc^\top y^*,
\]
where the first inequality follows from the fact that $y^*(e)=1\ \forall e\in E_1$, together with the fact that for every $e\in \vec{E}_2\cup \vec{E}_3$, $y^*(e)\in(0,1)$ and $y^*(e)$ is $1/k$-integral by Lemma \ref{lemma:1/k-integral}, which implies $y^*(e)\geq 1/k$. 
\end{proof}

\section{Min cost Symmetric Nowhere-Zero Flows}\label{sec:SWNZF}
In this section, we study $\operatorname{SWNZF}(k)$, the weighted nowhere-zero $k$-flow with symmetric costs. We prove the NP-hardness of $\operatorname{SWNZF}(k)$ in Section \ref{sec:SWNZF-hardness} and give a $3$-approximation algorithm for every $k\geq 6$ and $k=\infty$ in Section \ref{sec:SWNZF-approx}.

\subsection{Hardness of $\operatorname{SWNZF}(k)$}\label{sec:SWNZF-hardness}
We prove Theorem \ref{thm:complexity-SWNZF-k} in this subsection. We first prove the NP-hardness of $\operatorname{SWNZF}(\infty)$. We reduce from the NP-complete problem \emph{not-all-equal $3$-SAT (NAE3SAT)} \cite{schaefer1978complexity}: the input is a collection of $n$ variables and $m$ clauses, where each clause is a disjunction of $3$ variables or their negation. The goal is to determine if there is an assignment of Boolean values to variables such that the Boolean values assigned to the three variables in each clause are not all equal to each other (in other words, at least one is positive, and at least one is negative). Our construction is inspired by \cite{martinez2006complexity}.

\begin{theorem}\label{thm:complexity-SWNZF}
    $\operatorname{SWNZF}(\infty)$ for unit costs is NP-hard.
\end{theorem}
\begin{proof}
    Let the input instance to NAE3SAT consist of variables $x_1,...,x_n$ and clauses $C_1,...,C_m$. Suppose $x_i$ appears $a_i$ times positively and $a'_i$ times negatively for each $i\in [n]$ and let $d_i=\max\{a_i,a'_i\}$ for each $i\in [n]$. We construct an undirected graph $G=(V,E)$ in the following way. For each $i\in [n]$, construct a cycle $R_i$ of length $2d_i$ with vertices $u_1^i,u_2^i...,u_{2d_i}^i$ in cyclic order. For each clause $C_j$, there is a node $v_j$ corresponding to it. Let $(u^i_{2s-1},v_j)\in E$ if $x_i$ appears positively in $C_j$ for some $s\in [d_i]$; let $(u^i_{2s},v_j)\in E$ if $x_i$ appears negatively in $C_j$ for some $s\in [d_i]$. We choose $s$ in a way that each $v_j$ is adjacent to a distinct $u^i_t$, $t\in [2d_i]$. For those $u^i_t$ not adjacent to any $v_j$, we connect them to a special node $v_0$. Finally, we add an edge $(v_j,v_0)$ for each $j\in [m]$ (see Figure \ref{fig:complexity-SWNZF-1} (1) for details and Figure \ref{fig:complexity-SWNZF-3} for an example). All arc costs are unit, i.e., $c(e^+)=c(e^-)=1\ \forall e\in E$.
    \begin{figure}[htbp]
        \centering
        \includegraphics[width=\linewidth]{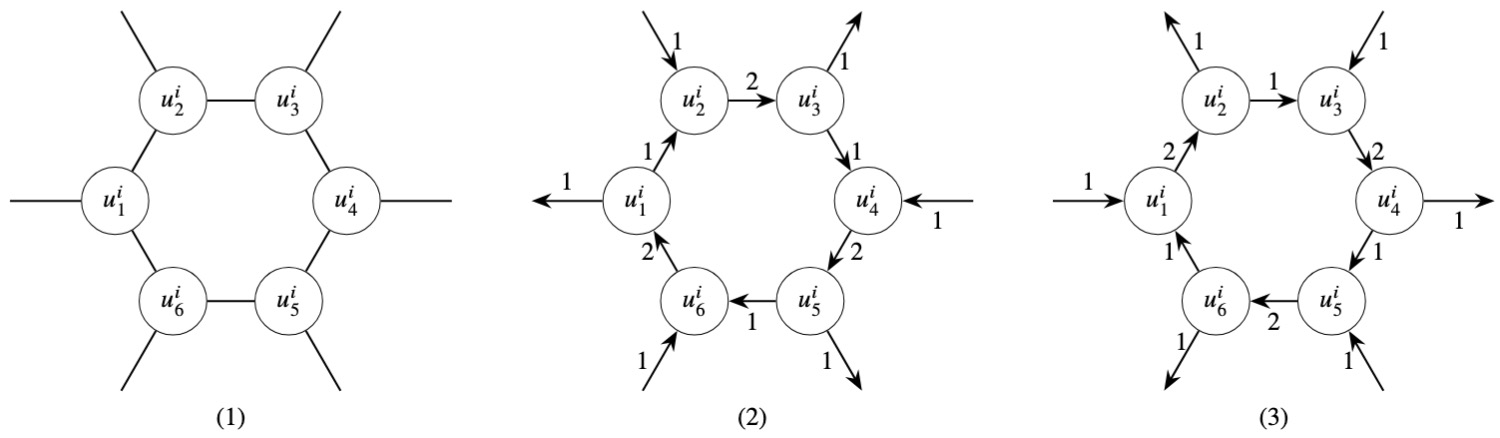}
        \caption{(1) Cycle $R_i$ corresponding to variable $x_i$ with $d_i=3$. (2) Part of a nowhere-zero $(\vec{E},f)$ when $x_i=1$. (3) Part of a nowhere-zero $(\vec{E},f)$ when $x_i=0$.}
        \label{fig:complexity-SWNZF-1}
    \end{figure}

        \begin{figure}[htbp]
    \centering
\includegraphics[width=0.8\linewidth]{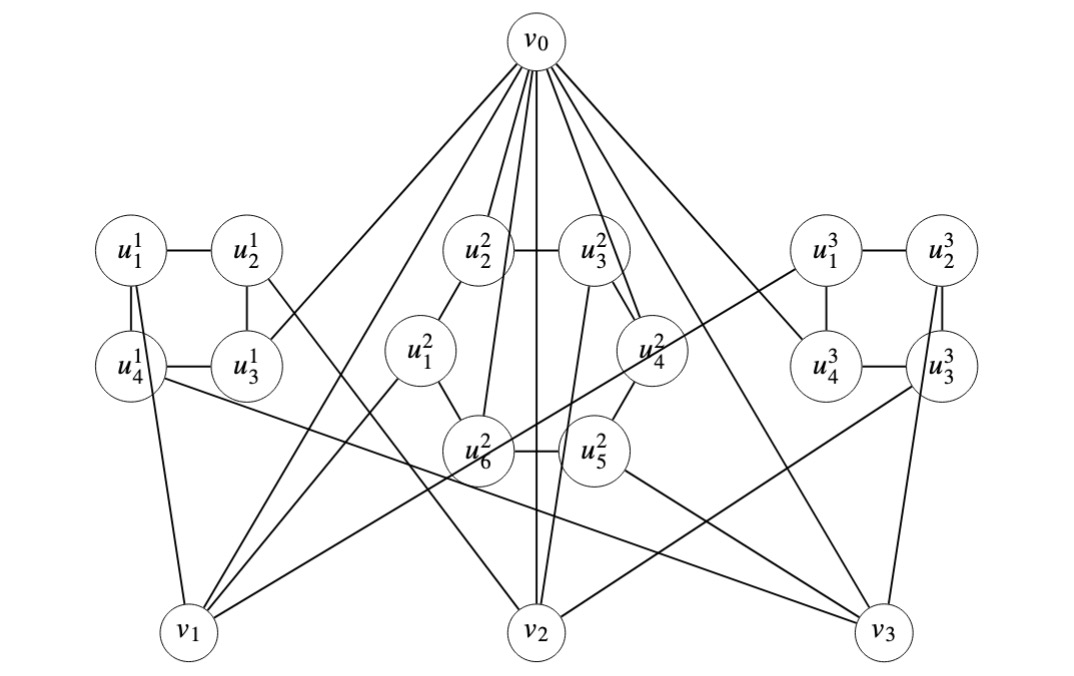}
        \caption{The graph $G$ for NAE3SAT instance $(x_1\vee x_2\vee x_3)\wedge (\bar{x}_1\vee x_2\vee x_3)\wedge (\bar{x}_1\vee x_2\vee \bar{x}_3)$.}
        \label{fig:complexity-SWNZF-3}
    \end{figure}
    We claim that the NAE3SAT instance is satisfiable if and only if the min total value of a nowhere-zero flow of $G$ equals $|E|+\sum_{i=1}^n d_i$. First, observe that the total value of a nowhere-zero flow is at least $|E|+\sum_{i=1}^n d_i$. Indeed, the degree of each $u^i_t$ is $3\ \forall i\in[n], t\in[2d_i]$. By flow conservation, at least one arc adjacent to $u^i_t$ has flow value at least $2$. Thus, the number of arcs having flow value at least $2$ is at least $\frac{1}{2}(\sum_{i=1}^n 2d_i)=\sum_{i=1}^n d_i$, where we divide by $2$ because $u_t^i$ and $u_{t+1}^i$ may share an arc of flow at least $2$ (we assume $u^i_{2d_i+1}:=u^i_1$). Therefore, every nowhere-zero flow has value at least $|E|+\sum_{i=1}^n d_i$.

        \begin{figure}[htbp]
    \centering
\includegraphics[width=0.8\linewidth]{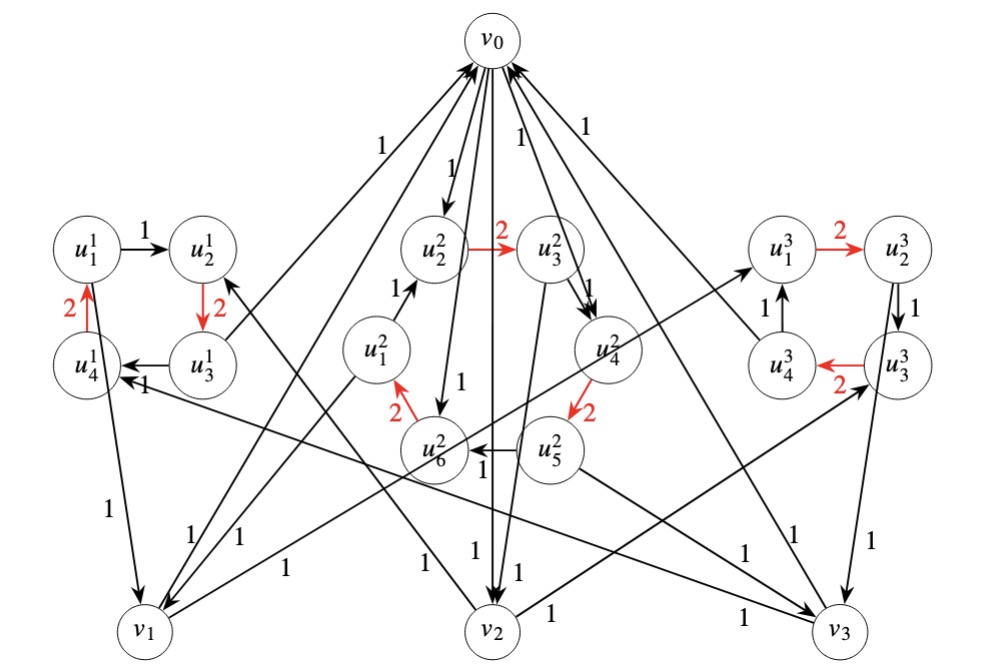}
        \caption{The nowhere-zero flow of total value $|E|+\sum_{i=1}^n d_i$ corresponding to a feasible assignment $x_1=1,x_2=1,x_3=0$ for NAE3SAT instance $(x_1\vee x_2\vee x_3)\wedge (\bar{x}_1\vee x_2\vee x_3)\wedge (\bar{x}_1\vee x_2\vee \bar{x}_3)$. The arcs of flow value $2$ are colored red, which are the even arcs of $R_1$, the even arcs of $R_2$, and the odd arcs of $R_3$.}
        \label{fig:complexity-SWNZF-4}
    \end{figure}
    Suppose $x_1,...,x_n$ is a feasible assignment such that each clause has at least one positive and one negative value. We construct a nowhere-zero flow $(\vec{E},f)$ of total value $|E|+\sum_{i=1}^n d_i$ in the following way. Orient each cycle $R_i$ as $u^i_1\rightarrow u^i_2\rightarrow\cdots \rightarrow u^i_{2d_i}\rightarrow u^i_1$. If $x_i=1$, orient $(u^i_{2s-1},v_j)\in \vec{E}$ and assign $f(u^i_{2s-1},v_j)=1$ for every $s\in[d_i], j\in \{0,1,...,m\}$ such that $(u^i_{2s},v_j)\in E$; orient $(v_j, u^i_{2s})\in \vec{E}$ and assign $f(v_j,u^i_{2s})=1$ for every $s\in[d_i], j\in \{0,1,...,m\}$ such that $(u^i_{2s},v_j)\in E$. Let $f(u^i_{2s-1},u^i_{2s})=1$, $f(u^i_{2s},u^i_{2s+1})=2\ \forall s\in[d_i]$ (see Figure \ref{fig:complexity-SWNZF-1} (2)). If $x_i=0$, orient $(v_j,u^i_{2s-1})\in \vec{E}$ and assign $f(v_j,u^i_{2s-1})=1$ for every  for every $s\in[d_i], j\in \{0,1,...,m\}$ such that $(u^i_{2s-1},v_j)\in E$; orient $(u^i_{2s},v_j)\in \vec{E}$ and assign $f(u^i_{2s},v_j)=1$ for every for every $s\in[d_i], j\in \{0,1,...,m\}$ such that $(u^i_{2s},v_j)\in E$. Let $f(u^i_{2s-1},u^i_{2s})=2$, $f(u^i_{2s},u^i_{2s+1})=1\ \forall s\in[d_i]$ (see Figure \ref{fig:complexity-SWNZF-1} (3)). Finally, for each $C_j$, there are either two positive and one negative values or one positive and two negative values. In the former case we orient $(v_j,v_0)\in \vec{E}$ and assign $f(v_j,v_0)=1$, while in the latter case we orient $(v_0,v_j)\in \vec{E}$ and assign $f(v_0,v_j)=1$ (see Figure \ref{fig:complexity-SWNZF-4}). This is indeed a nowhere-zero flow of total value $|E|+\sum_{i=1}^n d_i$, since the only arcs having flow value $2$ are the even arcs or the odd arcs in the cycles.

    We are left to prove the other direction: if the min total value of a nowhere-zero flow $(\vec{E},f)$ equals $|E|+\sum_{i=1}^nd_i$, then the instance is satisfiable. It follows from the proof of $c(f)\geq |E|+\sum_{i=1}^nd_i$ that the equality holds if and only if the flow values are all $1$ or $2$, and the value $2$ arcs are the even arcs or the odd arcs in each cycle $R_i$. Note that the arcs in $R_i$ have to be oriented consistently, since the flow values of arcs of the form $(u^i_t,v_j)$ or $(v_j,u^i_t)$ are all $1$. Moreover, the edges of the form $(u^i_t,v_j)$ are oriented in a way that either $(u^i_{2s-1},v_j)\in \vec{E}$, $(v_j,u^i_{2s})\in \vec{E}\ \forall s\in[d_i]$ or $(v_j,u^i_{2s-1})\in \vec{E}$, $(u^i_{2s},v_j)\in \vec{E}\ \forall s\in[d_i]$, depending on the value $2$ arcs are the even or odd arcs of $R_i$. In the former case we assign $x_i=1$, while in the latter case we assign $x_i=0$. This way, the value of $x_i$ is positive in $C_j$ if and only if $(u^i_t,v_j)\in \vec{E}$ for some $t\in [2d_i]$. This is indeed a feasible assignment, because for each clause $v_j$, $|\delta_{\vec{E}}^+(v_j)|=|\delta_{\vec{E}}^-(v_j)|=2$, since the arcs adjacent to $v_j$ all have flow values $1$. This implies that $C_j$ has at least one positive value and one negative value. This completes the proof.
\end{proof}
In fact, the nowhere-zero flows we use in the proof of Theorem \ref{thm:complexity-SWNZF} are nowhere-zero $3$-flows. Thus, the proof also implies that $\operatorname{SWNZF}(k)$ is NP-hard for every $k\geq 3$. Moreover, it is NP-hard to approximate $\operatorname{SWNZF}(k)$ within any finite factor for $k=3,4$, and for $k=5$ if Tutte's $5$-flow conjecture is false. Indeed, setting the costs $c=0$ recovers the feasibility problem of deciding whether a graph has a nowhere-zero $k$-flow for $k=3,4,5$ which are NP-complete.
Hence, Theorem \ref{thm:complexity-SWNZF-k} follows.

\subsection{Approximation algorithms}\label{sec:SWNZF-approx}
As we discussed earlier, every nowhere-zero $6$-flow is a $5$-approximation to $\operatorname{SWNZF}(k)$ for every finite integer $k\geq 6$ and for $k=\infty$.  Indeed, let $(\vec{E},f)$ be an arbitrary nowhere-zero $6$-flow, which is also a nowhere-zero $k$-flow. Let $\opt$ be a min cost nowhere-zero $k$-flow. Then,
\[
c(f)=\sum_{e\in\vec{E}} c(e)f(e)\leq \sum_{e\in\vec{E}} 5c(e)\leq 5c(\opt),
\]
where the first inequality follows from the fact that $f$ is a $6$-flow. The second inequality follows from the fact that $\opt$ is nowhere-zero.

We provide an improved $3$-approximation algorithm by finding a cheaper nowhere-zero $6$-flow. Given a nowhere-zero $6$-flow $(\vec{E},f)$, for an arbitrary directed cycle $C\in\vec{E}$, if $\sum_{e\in C} c(e)f(e)>\sum_{e\in C} c(e)(6-f(e))$, we can push $6$ units flow along the reverse direction of $C$. This way, we reduce the cost of the flow along the cycle while maintaining a nowhere-zero $6$-flow. This observation inspires the definition of locally optimal nowhere-zero $6$-flow: we say that a nowhere-zero $6$-flow is \emph{locally optimal} if for every directed cycle $C$ (wrt the flow),
\[
    \sum_{e\in C} c(e)f(e)\leq \sum_{e\in C} c(e)(6-f(e)),
\]
i.e.,
\begin{equation}\label{eq:local-opt}
    \sum_{e\in C} c(e)f(e)\leq 3\sum_{e\in C} c(e).
\end{equation}

We prove the following lemma.

\begin{lemma}\label{lemma:localopt}
    For every finite integer $k\geq 6$ and for $k=\infty$, every locally optimal nowhere-zero $6$-flow is a $3$-approximation to $\operatorname{SWNZF}(k)$.
\end{lemma}
\begin{proof}
    Let $(\vec{E},f)$ be a locally optimal nowhere-zero $6$-flow. We recall that a circulation $f$ can be decomposed into directed cycles $\mathscr{C}$, i.e., $f=\sum_{C\in\mathscr{C}} \chi(C)$, where $\chi(C)$ is the characteristic vector of $C$.
    Then,
    \[
    \begin{aligned}
        \sum_{e\in \vec{E}}c(e)f(e)
        =&\sum_{e\in \vec{E}}c(e)\sum_{C\in\mathscr{C}} \mathbf{1}\{e\in C\}
        =\sum_{C\in\mathscr{C}}\sum_{e\in C} c(e)\\
        \geq& \frac{1}{3}\sum_{C\in\mathscr{C}}\sum_{e\in C} c(e)f(e)
        = \frac{1}{3}\sum_{e\in\vec{E}}c(e)f(e)^2
        \geq \frac{1}{3}\frac{(\sum_{e\in\vec{E}}c(e)f(e))^2}{\sum_{e\in\vec{E}}c(e)},
    \end{aligned}
    \]
    where the first inequality follows from \eqref{eq:local-opt} and the second inequality follows from the Cauchy-Schwarz inequality. Let $\opt$ be the min cost nowhere-zero $k$-flow.
    It follows that
\[
c(f)=\sum_{e\in \vec{E}}c(e)f(e)\leq 3 \sum_{e\in\vec{E}}c(e)\leq 3c(\opt). 
\]
\end{proof}

Next, we show that a local optimum exists by giving a pseudo-polynomial algorithm to find it. 
\begin{prop}
\label{prop:local}
    A locally optimal nowhere-zero $6$-flow exists and there is a pseudo-polynomial algorithm to find it.
\end{prop}
\begin{proof}
    Start from an arbitrary nowhere-zero $6$-flow $(\vec{E},f)$, which can be found in polynomial time \cite{younger1983integer}. Check if there is a directed cycle $C\subseteq \vec{E}$ violating \eqref{eq:local-opt}, i.e., $\sum_{e\in C} c(e)(3-f(e))<0$. This can be done using an algorithm such as the Bellman-Ford algorithm that detects negative cycles in a weighted digraph $\vec{E}$ with weights $c(e)(3-f(e))\ \forall e\in\vec{E}$. If so, we push $6$ units flow along the reverse direction of $C$. This way, we obtain a new nowhere-zero $6$-flow, and then we repeat the above step.

    We show that the algorithm terminates. In each step, the total cost of the flow $\sum_{e\in \vec{E}}c(e)f(e)$ reduces by $\sum_{e\in C} c(e)\big(f(e)-(6-f(e))\big)=\sum_{e\in C} 2c(e)(f(e)-3)>0$. Since $c\in\Z_{\geq 0}$, it always reduces by an integer at least $1$. Thus, the algorithm terminates in at most $\sum_{e\in \vec{E}}c(e)f(e)\leq 5\sum_{e\in \vec{E}}c(e)$ steps.
\end{proof}

We note that to arrive at a locally optimal nowhere-zero 6-flow $f'$ starting with $f$ in the above proof, we add a collection of cycles pushing $6$ units flow. This sum of cycles each pushing $6$ units flow gives a feasible solution to a circulation problem where all the nonzero flows in the circulation have value $6$. We formulate the task of finding such a circulation of minimum cost to obtain a strongly polynomial time algorithm for finding a locally optimal nowhere-zero $6$-flow.

% In particular, to achieve a strongly polynomial time algorithm finding a local optimum, we convert the problem into a min cost circulation problem~\cite{KorteVygen-book} \rnote{Pls add this or another reference}. 
Let $(\vec{E},f)$ be an arbitrary nowhere-zero $6$-flow. We are now allowed to flip the direction of some arcs of $\vec{E}$ while maintaining a nowhere-zero $6$-flow. We restrict ourselves to always pushing $6$ units flow along the reverse direction of some directed cycle $C$. This way, for every arc $e\in \vec{E}$ we flip, the flow value $f(e)$ assigned to it becomes $6-f(e)$. Thus, the change in the cost by flipping $e$ becomes $c(e)\big((6-f(e))-f(e)\big)=2c(e)
(3-f(e))$. This motivates us to formulate the following min cost circulation problem.

Let $D=(V,\cev{E})$ be the directed graph consisting of arcs in $\cev{E}:=(\vec{E})^{-1}$. For each $e\in \cev{E}$, let cost $c'(e):=c(e)(3-f(e^{-1}))$ and capacity $l(e)=0,\ u(e)=1$. An integral $g:\cev{E}\rightarrow \Z$ is feasible if it is a circulation, i.e., $g(\delta_{\cev{E}}^+(v))=g(\delta_{\cev{E}}^-(v))\ \forall v\in V$, and it satisfies the capacity constraints $l(e)\leq g(e)\leq u(e)\ \forall e\in \cev{E}$. Let $g$ be a min cost integral circulation of the instance. Return $f+6g$ (notice that $g$ is a $2$-flow, which allows us to define $f+6g$ the same way as in Section \ref{sec:basic}). In the lemma below, we show that $f+6g$ is a locally optimal nowhere-zero $6$-flow.

\begin{lemma}\label{lemma:localopt+mincostflow}
    Let $(\vec{E},f)$ be a nowhere-zero $6$-flow and $D=(V,\cev{E})$ be a digraph with costs $c'(e)=c(e)(3-f(e))\ \forall e\in\cev{E}$ and capacities $l=0,u=1$. Then, $g:\cev{E}\rightarrow \Z$ is a min cost circulation if and only if $f+6g$ is a locally optimal nowhere-zero $6$-flow.
\end{lemma}
\begin{proof}
    First, we claim that for an arbitrary $g:\cev{E}\rightarrow\Z$, $g$ is feasible, i.e., a circulation obeying capacity constraints $0\leq g(e)\leq 1\ \forall e\in \cev{E}$, if and only if $f+6g$ is a nowhere-zero $6$-flow. If $g$ is feasible, then $g(e)\in\{0,1\}\ \forall e\in\cev{E}$. For an arbitrary $e\in \cev{E}$ with $g(e)=1$, one has $1=(-5)+6\leq f(e)+6g(e)\leq (-1)+6=5$, where we use the facts that $f(e)=-f(e^{-1})$ and $1\leq f(e^{-1})\leq 5$. Otherwise, $g(e)=0$ and thus $f(e)+6g(e)=f(e)$, which implies $-5\leq f(e)+6g(e) \leq -1$. This implies that $f+6g$ is a nowhere-zero $6$-flow. Conversely, if $g$ is not feasible, we show that  $f+6g$ is not a nowhere-zero $6$-flow. If $g$ is not a circulation, then $f+6g$ is not a circulation and thus not a nowhere-zero flow. If $g$ violates capacity constraints, then there exists $e\in \cev{E}$ such that $g(e)\geq 2$ or $g(e)\leq -1$, then $f(e)+6g(e)\geq (-5)+6\cdot 2=7$ or $f(e)+6g(e)\leq (-1)+6\cdot (-1)=-7$, respectively, which implies that $f+6g$ is not a nowhere-zero $6$-flow.

For convenience, from now on we will work with the following variant of the problem. For an arbitrary flow $g:\cev{E}\rightarrow \Z$, we extend its domain to $E^+\cup E^-$ 
% \rnote{$\vec{E} \cup \cev{E}$?}
by letting $g(e):=-g(e^{-1})\ \forall e\in \vec{E}$. For each $e\in \vec{E}$, let cost $c'(e):=-c'(e^{-1})=-c(e)(3-f(e))$. 
Now, the cost of $g$ becomes $\sum_{e\in \vec{E}\cup \cev{E}}c(e)g(e)=2\sum_{e\in \cev{E}}c(e)g(e)$. Since for circulations $g$ and $g'$, $\sum_{e\in \vec{E}\cup \cev{E}}c'(e)g(e)\geq \sum_{e\in \vec{E}\cup \cev{E}}c'(e)g'(e)$ if and only if $\sum_{e\in \cev{E}}c(e)g(e)\geq \sum_{e\in \cev{E}}c(e)g'(e)$, the min cost circulation stay unchanged. 
% \rnote{It is somewhat unclear to me where you actually use this generalization below. }
% \rnote{Yes, I see it now in the equalities below.}

    Now, we prove the ``only if" direction. Suppose $g$ is a min cost circulation. Then, $g(e)\in\{0,1\}\ \forall e\in \cev{E}$. Moreover, let $\vec{E}':=\supp^+(f+6g)$ be the orientation associated with $f+6g$, i.e., $\{e\in E^+\cup E^-: (f+6g)(e)>0\}$. It follows from Proposition \ref{prop:sum-flows} that for an arbitrary $e\in \vec{E}'$, $e\in \cev{E}$ if and only if $g(e)=1$. For the sake of contradiction assume that $f+6g$ is not a local optimum. Then, by \eqref{eq:local-opt}, there exists some cycle $C\subseteq \vec{E}'$, satisfying 
    \begin{equation}\label{eq:negative_cycle}
        \begin{aligned}
            0>&\sum_{e\in C} c(e)(3-(f+6g)(e))=\sum_{e\in C\cap \vec{E}} c(e)(3-f(e))+\sum_{e\in C\cap \cev{E}} c(e)(3-(-f(e^{-1})+6))\\
            =&\sum_{e\in C\cap \vec{E}} c(e)(3-f(e))+\sum_{e\in C\cap \cev{E}} c(e)(f(e^{-1})-3)=-\sum_{e\in C}c'(e).
        \end{aligned}
    \end{equation}
    Let $\chi_C$ be the indicator vector of $C$ defined as $\chi_C(e)=1\ \forall e\in C$, $\chi_C(e)=-1\ \forall e\in C^{-1}$, and $\chi_C(e)=0, \text{otherwise}$.
    Let $g'=g-\chi_C$. Observe that $f+6g'=f+6g-6\chi_C$ is a nowhere-zero $6$-flow since $C\subseteq\vec{E}'$. Thus, it follows from the claim at the beginning of the proof that $g'$ is feasible. However, $c'(g')-c'(g)=-c'(C)<0$, a contradiction to the optimality of $g$. 
    % We prove that $g'$ is also a feasible flow of $D$ satisfying capacity constraints. To see this, when $\chi_C(e)=1$, $e\in \vec{E}'$, and thus $e\in \cev{E}$ if and only if $g(e)=1$. If $e\in \cev{E}$, $g'(e)=g(e)-\chi_C(e)=1-1=0$; if $e\in \vec{E}$, $g'(e)=g(e)-\chi_C(e)=0-1=-1$.

    We then prove ``if" direction. Suppose $f+6g$ is a locally optimal nowhere-zero $6$-flow. By the claim, $g$ is feasible. For the sake of contradiction, assume $g$ is not a min cost circulation. Let $g^*$ be a min cost circulation. The difference $g-g^*$ is also a circulation, and thus can be decomposed into cycles. Since $c'(g)>c'(g^*)$, there exists some cycle $C$ with $c'(C)>0$ such that $g':=g-\chi_C$ is also a feasible circulation. Thus, it follows from the claim that $f+6g-6\chi_C=f+6g'$ is also a nowhere-zero $6$-flow. This implies $C\subseteq \vec{E}'$. By \eqref{eq:negative_cycle}, $\sum_{e\in C} c(e)(3-(f+6g)(e))=-c'(C)<0$. This contradicts the condition \eqref{eq:local-opt} that has to be satisfied when $f+6g$ is a local optimum.
\end{proof}

Lemmas \ref{lemma:localopt} and \ref{lemma:localopt+mincostflow} together complete the proof of Theorem \ref{thm:approx-SWNZF-k}.
% \begin{theorem}\label{thm:approx-SWNZF}
%     There is a $3$-approximation algorithm for $\operatorname{SWNZF}(\infty)$.
% \end{theorem}

\begin{proof}[Proof of Theorem \ref{thm:approx-SWNZF-k}]
     Start from an arbitrary nowhere-zero $6$-flow, which can be computed in polynomial time using Younger's algorithm~\cite{younger1983integer}. Lemma \ref{lemma:localopt+mincostflow} implies that a locally optimal nowhere-zero $6$-flow can be found in strongly polynomial time using min cost circulation algorithms (see e.g. \cite{ford2015flows,chen2022maximum}). Lemma~\ref{lemma:localopt} implies that this is a $3$-approximation to $\operatorname{SWNZF}(k)\ \forall k\geq 6$ and $k=\infty$.
\end{proof}

\begin{remark}
    The approximation ratio $3$ is tight for this algorithm. To see this, consider the entire graph to be a cycle $C$, with unit costs. An optimal nowhere-zero flow is $(\vec{C},f)$ with $f(e)=1\ \forall e\in\vec{C}$, whose cost is $|C|$. However, if our algorithm starts from the nowhere-zero $6$-flow $(\vec{C},f')$ with $f'(e)=3\ \forall e\in\vec{C}$, this is already a local optimum, so we only obtain a nowhere-zero flow of cost $3|C|$.
\end{remark}

\section{Conclusion}
We conclude with some interesting directions for future research. Firstly, is it possible to obtain a bicriteria approximation for min-cost well-balanced orientations? In particular, given an undirected graph $G=(V, E)$, is it possible to construct an orientation $\vec{E}$ such that (i) the cost of the orientation is at most a constant factor of a min-cost well-balanced orientation and (ii) the orientation $\vec{E}$ is constant-approximately well-balanced, i.e., $\lambda_{\vec{E}}(u,v)\ge \rounddown{\lambda_E(u,v)/\rho}$ for some constant $\rho$? Secondly, can we improve on the $(6,6)$-bicriteria approximation for $\operatorname{WNZF}(k)$, on the $(k,6)$-bicriteria approximation for $\operatorname{WCBO}(k)$, or on the $3$-factor approximation for $\operatorname{SWNZF}(k)$ for finite integer $k\ge 6$?

\paragraph*{Acknowledgements} The first two authors would like to thank Eml\'ekt\'abla workshop in July 2024 and ICERM reunion workshop on \emph{Discrete Optimization: Mathematics, Algorithms, and Computation} in August 2024, where this work was initiated. We would like to thank Krist\'of B\'erczi and G\'erard Cornu\'ejols for initial discussions.

% \newpage

\bibliographystyle{plain} % We choose the "plain" reference style
\bibliography{reference}

\end{document}